\pdfoutput=1

\documentclass[twocolumn,a4paper,9pt]{extarticle}

\usepackage[a4paper,textwidth=504pt,textheight=682pt,centering]{geometry}
\usepackage{newtxtext,newtxmath}
\usepackage[T1]{fontenc}
\usepackage{natbib}
\bibpunct[, ]{(}{)}{;}{a}{}{,}

\InputIfFileExists{cil-preamble.tex}{}{\input{../cil-preamble}}
\newcommand{\CILTITLE}{CLASSO: a practical primal path-following
algorithm for the constrained LASSO}

\newcommand{\CILKEYWORDS}{Constrained lasso; solution path; active-set
method; parametric quadratic programming; Karush--Kuhn--Tucker
certificate; formal verification}

\newcommand{\REPOFOOTNOTE}{\footnote{Release commit \texttt{ff9681a}, tagged \texttt{classo-paper-2026-10},
available from the authors on reasonable request pending public
release.}}
\newcommand{\REPOWHERE}{the released code}

\newcommand{\CILBIBLIOGRAPHY}{%
  \bibliographystyle{tfcad}%
  \small
  \bibliography{../cil}}

\title{\bfseries \CILTITLE}

\author{%
Michael Hobson, Luke Butcher and Lindley Lentati\\[3pt]
\normalsize Cambridge Inference Ltd, 95 Hills Road, Cambridge CB2 1AW, UK\\[1pt]
\normalsize\texttt{\{mph,\,lmb,\,ltl\}@cambridgeinference.com}}

\date{}

\begin{document}

\makeatletter
\twocolumn[%
  \begin{@twocolumnfalse}
  \maketitle
  \begin{center}
  \begin{minipage}{0.88\textwidth}
  \begin{abstract}
  \noindent
  We present CLASSO, a path-following method for the LASSO subject to
general linear equality and inequality constraints. The local
direction is obtained from a quadratic programme on an augmented
active set, which determines both the coefficients that leave zero
and the weakly active inequalities that remain binding. For a
positive-definite Hessian and a nonempty feasible set, this
programme has a unique primal minimiser, even when the binding
constraint rows are dependent, and gives the local derivative of
the solution path. The construction recovers the LASSO-modified
LARS direction without constraints, and over the full cone gives
the segment leaving the origin. We prove that the primal path is
piecewise affine with finitely many kinks, without a constraint
qualification or unique multipliers. A finite path does not
imply a finite traversal: re-certifying at each kink can fail to
terminate. CLASSO instead follows the
minimum-norm certificate along a lexicographic rate; we prove that
this rule reaches $\lambda=0$ in exact arithmetic within an
explicit bound and traverses the whole path. These results are
machine-checked in Lean~4, and the floating-point implementation
agrees with an exact rational one on all $210$ small test
problems. We assess the implementation on a $39$-case
synthetic suite and six constrained versions of the diabetes
benchmark. On the $31$ nontrivial synthetic cases with $n\ge p$,
sampled relative objective differences from numerical reference
solves are below $10^{-9}$ in magnitude; unnormalised endpoint
differences reach $1.5\times10^{-7}$ in magnitude. For $n<p$, matched differences range
from $-3.6\times10^{-7}$ to $2\times10^{-9}$ across five ridge
strengths. Across the synthetic and diabetes benchmarks, checks
of all $580$ recorded segments found no missed interior event and
all checked KKT residuals within the stated tolerances.

  \smallskip
  \noindent\textbf{Keywords:} \CILKEYWORDS
  \end{abstract}
  \end{minipage}
  \end{center}
  \bigskip
  \end{@twocolumnfalse}%
]
\makeatother

\section{Introduction}\label{sec:intro}

The LASSO \citep{Tibshirani1996} estimates the coefficients of a
linear regression model by minimising a least-squares loss together
with an $L_1$ penalty, which encourages sparsity by setting some
coefficients to zero. In many applications, however, the
coefficients must also satisfy constraints arising from the
interpretation of the model or from prior information. These may
include non-negativity of concentrations or mixture weights,
sum-to-zero conditions in compositional regression, or more general
linear relationships between the coefficients. One is thus led to
the constrained LASSO, in which the usual objective is minimised
subject to linear equality and inequality constraints.

Provided the Hessian is positive definite and the feasible set is
nonempty, the constrained problem has a unique solution at each
value of the regularisation parameter $\lambda$, and the solution
path is continuous and piecewise affine. One may approximate this
path by solving a quadratic programme (QP) on a discrete
$\lambda$-grid, but the resulting resolution depends on the grid
spacing, and the kink locations are not determined directly.
Alternatively, a path-following method computes the successive
affine segments and the values of $\lambda$ at which their
directions change. For the unconstrained LASSO, this is the
approach adopted by the LASSO-modified LARS algorithm of
\citet{Efron2004}.

A path-following method for the constrained LASSO was proposed by
\citet{Gaines2018}, who obtain the local direction from a
saddle-point linear system on the active coefficients and binding
constraints. The invertibility of this system requires the
binding-constraint matrix restricted to the active coefficients to
have full row rank. This condition, stated by \citet{Gaines2018}
and standard in saddle-point linear algebra \citep{Benzi2005},
fails frequently in the examples considered here. In particular,
a binding non-negativity constraint is supported on a coefficient
that is zero, so its restriction to the non-zero coefficients is
identically zero.

The difficulty is not resolved merely by replacing the inverse
with a rank-aware linear solve. The equations may remain
consistent, but imposing every binding row as an equality can
prevent a coefficient from leaving zero when the corresponding
inequality should become slack. Thus, before computing the
direction by linear algebra, one must determine which boundary
coordinates move and which weakly active inequalities remain
binding. Our comparison with the Gaines construction concerns
this local binding-block calculation; it does not reproduce the
complete algorithm of \citet{Gaines2018}, including its
conditional rule for admitting boundary coordinates.

In this paper we present CLASSO, which determines these choices
by solving a directional QP on an augmented active set. The QP
has a unique primal minimiser even when the binding rows are
linearly dependent, although its multipliers need not be unique.
We show that, at an exact KKT point, this minimiser gives the
local direction of the constrained-LASSO path. Over the full cone
the same QP gives the segment leaving the origin, so no separate
entry construction is needed. We also give a lexicographic
initialisation at the largest value of $\lambda$, and the
two-sided subgradient test required for coordinates held at zero
within the augmented active set. In the unconstrained setting the
construction recovers the LASSO-modified LARS direction.

An important practical consideration is that errors in the local
direction affect the predicted event times and hence the state
used on the next segment. We therefore use the directional QP
to identify a candidate release decision, refine this decision by
KKT sign tests, and recover the direction by solving the linear
system for the resulting state. At each recorded point the
coefficients and the certificate are recomputed rather than
propagated, so that error does not accumulate. These operations are exact
linear-algebraic constructions in the mathematical algorithm;
their numerical implementation remains subject to rounding,
conditioning and the tolerances described below.

The local results must be distinguished from a guarantee that a
traversal reaches the end of the path. We prove that the primal
path has finitely many affine pieces, without requiring a
constraint qualification, strict complementarity or unique
multipliers. Finitely many pieces do not by themselves make a
traversal finite: the multipliers are generally not unique, and a
rule that recovers a fresh certificate at each kink can record
infinitely many events along a single primal segment, as we show
by counterexamples. CLASSO instead follows a canonical
certificate, the minimum-norm valid multiplier, along a
lexicographic rate, and records a further event when that
certificate changes face. We prove that this rule reaches
$\lambda=0$ in exact arithmetic within an explicit, if
exponential, number of steps, and that, started from the
lexicographic initialisation, it traverses the whole path. These
results and the counterexamples are machine-checked in Lean~4.
That the floating-point implementation follows the exact rule is
assessed numerically: its propagated certificate is checked against
the canonical one at every step, and its output against an exact
rational implementation.

We assess the numerical implementation on a $39$-case synthetic
suite and six constrained versions of the diabetes benchmark.
For the synthetic problems with $n\ge p$, the computed objective
exceeds a numerical reference solve by at most $3.9\times10^{-10}$
relative at any sampled $\lambda$ along the path. At $\lambda=0$, signed unnormalised differences
reach $-1.5\times10^{-7}$, with CLASSO returning the lower computed
objective. The six
$n<p$ problems are solved after ridge augmentation, and comparison
with numerical reference solves of the same ridge-augmented objective
gives matched differences ranging from $-3.6\times10^{-7}$ to
$2\times10^{-9}$ over five prescribed ridge strengths.
The effect of the ridge on the original problem is reported
separately. We also verify the KKT conditions throughout the $580$
recorded affine segments, rather than only at sampled points. That
verification found no missed interior event.
On the diabetes examples, CLASSO uses $30$--$84$ times
fewer LP and QP solves than a $1000$-point sweep, with
wall-clock speedups of $4$--$20$; a separate study
examines scaling up to $p=5000$.

The paper is organised as follows. Section~\ref{sec:related}
reviews related methods, and Section~\ref{sec:problem} introduces
the problem and notation. Sections~\ref{sec:gaines_prior}
and~\ref{sec:new_approach} describe the previous linear-system
construction and the CLASSO algorithm, respectively, the latter
comprising the directional QP, the certificate rule, the
initialisation and a summary in pseudocode. The theoretical
results and their formal verification are presented in
Section~\ref{sec:theory}, and the numerical implementation in
Section~\ref{sec:implementation}. Sections~\ref{sec:synthetic}
and~\ref{sec:diabetes_app} report the experiments, and
Section~\ref{sec:conc} gives our conclusions. Appendix~\ref{sec:safeguards}
describes the finite-precision failure mechanisms motivating the
implementation, while Appendix~\ref{sec:branch_uniqueness}
proves the local continuation and safe-hold propositions.

\section{Related work}\label{sec:related}

The piecewise-affine structure of the constrained-LASSO path is
closely related to classical parametric quadratic programming
\citep{Best1996,Tondel2003}. Within a region where the relevant
active constraints remain fixed, the optimality conditions give
an affine dependence on the parameter; changes of region give the
breakpoints. This observation also underlies explicit
model-predictive control. Our contribution is a directional-QP
construction for the constrained LASSO and its numerical
implementation, rather than the piecewise-affine structure itself.

A recent connection with portfolio optimisation is provided by
\citet{Schmelzer2026}, who relate the constrained LASSO to the
Critical Line Algorithm of \citet{Markowitz1956}. Under
$\bm{\Sigma}=\bm{X}^{\rm t}\bm{X}$ and
$\bm{\mu}=\bm{X}^{\rm t}\bm{y}$, the objectives agree up to an
additive constant, with the parameters related by
$c(\lambda)=\lVert\bm{\beta}(\lambda)\rVert_1$. Breakpoints
correspond one-to-one on intervals where this map is strictly
decreasing. Their four event families are a free coordinate
reaching a bound, an inequality becoming tight, a blocked
coordinate releasing, and a binding inequality's multiplier
reaching zero. These correspond to our events \textbf{E1},
\textbf{E3}, \textbf{E2} and \textbf{E4}. Their coordinate-release
test monitors both margins of $[-\lambda,\lambda]$; our
\textbf{E5} denotes the second margin when it is approached by a
coordinate held within the augmented active set. Our \textbf{E6},
a change of face of the certificate, has no counterpart there: it
changes the multipliers' rate but not the primal path.

The assumptions in the two treatments differ.
\citet{Schmelzer2026} permit $\bm{H}\succeq\bm{0}$ provided the
reduced Hessian is positive definite on each visited free block,
which can include unregularised problems with $n<p$. Here we
assume $\bm{H}\succ\bm{0}$, using ridge augmentation when needed.
Their curve correspondence assumes general position and LICQ
at the visited active sets, whereas our structural theorem and
local continuation result do not require a constraint
qualification. The distinction is relevant when binding rows are
dependent, as occurs frequently under the sign constraints
examined in Section~\ref{sec:gaines_comparison}.

For unconstrained regression, the LASSO-modified LARS algorithm
of \citet{Efron2004} traces the full path by monitoring coefficient
entries and departures. Related methods for the generalised
LASSO \citep{Tibshirani2011,Tibshirani2012}, implemented in
\texttt{genlasso} \citep{Arnold2016}, use the dual formulation
of a penalty $\lVert\bm{D}\bm{\beta}\rVert_1$. Their formulation
differs from one in which general equality and inequality
constraints are imposed directly on $\bm{\beta}$.

Two earlier path-following constructions address the constrained
LASSO directly. \citet{He2011} uses the decomposition
$\bm{\beta}=\bm{\beta}_+-\bm{\beta}_-$, with
$\bm{\beta}_\pm\ge\bm{0}$, thereby doubling the number of
coefficient variables; the treatment is given for $n\ge p$.
\citet{Gaines2018} works in the native $p$-dimensional parameter
space, permits ridge augmentation for $n<p$, and obtains the local
direction from a saddle-point KKT system. CLASSO also works with
the native coefficients, but determines the direction through a
QP on the augmented active set, whose primal well-posedness does
not require independent binding rows.

Several methods address related formulations or compute solutions
on a grid. The equality-constrained solver of \citet{Tran2024}
combines active-set screening with an ADMM/Newton inner solver at
prescribed values of $\lambda$. MM-DUST \citep{Chen2025} is a
majorisation--minimisation dual stagewise method for the
generalised LASSO, which includes constrained-LASSO problems and
extends beyond Gaussian regression, but advances in small fixed
steps rather than directly between exact breakpoints.
The package \texttt{ConstrainedLasso.jl} \citep{Gaines2018}
provides a constrained path-following implementation using
\texttt{Convex.jl}, while \texttt{c-lasso}
\citep{Simpson2021} supports equality constraints and several
losses, including square, Huber and classification losses.
Its documented constraints have the form $\bm C\bm\beta=\bm0$;
this should be distinguished from general affine equalities.
It does not address inequality constraints.

The \emph{comlasso} algorithm of \citet{Jeon2020} exploits
block-disjoint sum-to-zero constraints in compositional
regression, using a separate dual-feasibility check for each
block. Its developed formulation does not include inequalities
or general couplings between blocks, although the authors
discuss a possible extension using linear programming.
EPSODE \citep{Zhou2014} follows regularisation paths through an
ordinary differential equation, allowing equality and inequality
violations to be penalised through
$\rho\|\bm{V}\bm{\beta}-\bm{d}\|_1$ and
$\rho\|\bm{W}\bm{\beta}-\bm{e}\|_+$. Under the exact-penalty conditions of that work, the hard
constrained solution is attained at a sufficiently large finite
$\rho$. The path parameter nevertheless differs from the
$\lambda$ considered here, for which the linear constraints
are imposed throughout. The penalised-and-constrained regression method of
\citet{James2013} operates on a tuning-parameter grid;
\citet{Hu2015} gives a further related constrained
variable-selection construction.

For a fixed $\lambda$, proximal and operator-splitting methods
provide an alternative to path-following. Examples include
ISTA and FISTA \citep{Beck2009}, ADMM \citep{Boyd2011}, and
the methods reviewed by \citet{Parikh2014}. Suitable splittings
separate the loss, penalty and constraints into tractable
subproblems, making these approaches attractive when only a
small number of fits is required or the coefficient dimension
makes dense factorisations expensive. Although such solvers can
supply optimality residuals and dual information at the requested
values of $\lambda$, they do not in themselves identify the
complete breakpoint sequence.

The appropriate comparison therefore depends on the required
output. A grid of independently solved problems may suffice
for model selection or an approximate display of the path,
whereas a path-following method also determines its affine
segments and active-set transitions. CLASSO is intended for
the latter task at moderate dimensions. Its cost depends on the
number of iterations actually performed, which need not equal
the number of primal kinks in degenerate cases;
Section~\ref{sec:termination} bounds it.
Table~\ref{tab:related_work_summary} summarises the formulations
and outputs of the most closely related methods.

\begin{table*}
\centering
\caption{Scope summary of constrained-LASSO path-following
methods most directly comparable to CLASSO. ``eq'' / ``ineq''
denote support for linear equality / inequality constraints;
\emph{hom.} denotes general homogeneous equalities, as documented for
\texttt{c-lasso}, rather than arbitrary non-zero right-hand sides;
``general linear'' indicates that arbitrary linear rows are
admitted (vs constraints of a structured form, e.g. block
sum-to-zero); ``exact path'' denotes a piecewise-affine kink
sequence rather than a grid approximation; ``$n<p$'' denotes
support for the underdetermined regime (with a small ridge
augmentation where applicable); ``dual cert.'' indicates that
KKT multipliers are recovered at every kink. A non-checkmark
entry such as ``$D$-form'', ``sum-zero'', ``block'' or ``penalty''
indicates that the corresponding constraint type is supported only
in a restricted, structured, or penalty-form sense rather than as
a hard linear constraint at every $\lambda$.}
\label{tab:related_work_summary}
\smallskip
\CILFITWIDTH{%
\begin{tabular}{lcccccc}
\toprule
Method & eq & ineq & gen.\ lin. & exact path & $n<p$ & dual cert. \\
\midrule
LARS / LASSO-modified LARS \citep{Efron2004} & --- & --- & --- & \checkmark & \checkmark & \checkmark \\
genlasso \citep{Tibshirani2011,Arnold2016}   & $D$-form & --- & structured & \checkmark & ridge & \checkmark \\
\citet{He2011}                                & \checkmark & \checkmark & \checkmark & \checkmark & --- & \checkmark \\
\citet{Gaines2018}                            & \checkmark & \checkmark & \checkmark & \checkmark$^\ast$ & ridge & \checkmark \\
comlasso \citep{Jeon2020}                     & sum-zero & --- & block & \checkmark & \checkmark & \checkmark \\
c-lasso \citep{Simpson2021}                   & hom. & --- & hom. & \checkmark & \checkmark & \checkmark \\
EPSODE \citep{Zhou2014}                       & penalty & penalty & \checkmark & ODE & \checkmark & --- \\
\citet{Tran2024}                              & \checkmark & --- & \checkmark & grid & \checkmark & --- \\
\textbf{CLASSO (this paper)}                  & \checkmark & \checkmark & \checkmark & \checkmark$^\dagger$ & ridge & \checkmark \\
\bottomrule
\end{tabular}}

\medskip
\footnotesize

$^\ast$The explicit inverse in the Gaines direction step requires
the rank condition discussed in
Section~\ref{sec:gaines_prior}. The comparison in
Section~\ref{sec:gaines_comparison} concerns that local
linear-system construction, rather than a reproduction of the
complete Gaines algorithm.

$^\dagger$CLASSO constructs affine segments rather than a grid
approximation; its traversal is proved to terminate in exact
arithmetic (Section~\ref{sec:termination}). For $n<p$, it
traces the ridge-augmented problem. The entries in this table
describe formulations and outputs, not guarantees of numerical
accuracy or termination for every admissible input.

\end{table*}

\section{Problem formulation}\label{sec:problem}

We consider the constrained optimisation problem
\begin{equation}
\min_{\bm{\beta}} \tfrac{1}{2}\lVert \bm{y}-\bm{X}\bm{\beta} \rVert_2^2
+ \lambda \lVert \bm{\beta} \rVert_1\quad\mbox{subject
to}\quad \bm{A\beta} \preceq  \bm{b}, 
\label{eq:CLASSOprobdef}
\end{equation}
where $\bm{y}$ is an $n$-vector, $\bm{X}$ an $n\times p$ design
matrix, $\bm{\beta}$ a $p$-vector of coefficients, and
$\lambda\ge0$ the regularisation parameter. The constraint
matrix $\bm{A}$ has size $c\times p$ and $\bm{b}$ is a
$c$-vector. We use $\preceq$ to indicate that each row is
specified either as an equality or as a less-than-or-equal
inequality; $\le$ denotes inequalities throughout.

Unless otherwise stated, $\bm{X}$ is assumed to have full
column rank, which requires $n\ge p$, and the feasible set is
nonempty. For $n<p$, we follow the ridge-augmentation approach
discussed in the literature
\citep{Tibshirani2011,Tibshirani2012,Hu2015,Arnold2016,Gaines2018}
and add $\tfrac12\epsilon\lVert\bm{\beta}\rVert_2^2$ to the
objective, with $\epsilon=10^{-4}\max_i H_{ii}$ unless stated
otherwise. This makes the Hessian positive definite, but also
changes the optimisation problem. Consequently, we distinguish
the accuracy of the computed ridge-augmented path from its
difference from the unregularised solution. The same ridge is
used in initialisation, direction calculation, certification,
event detection and the matched reference solves.

The augmented problem may again be written in the form
(\ref{eq:CLASSOprobdef}) by appending $p$ zero observations
to $\bm{y}$ and the rows of $\sqrt{\epsilon}\bm{I}_p$ to
$\bm{X}$. In subsequent expressions, $\bm{X}$ and the
corresponding Hessian are understood to include this
augmentation where applicable.

The Karush--Kuhn--Tucker (KKT) conditions
\citep{Karush1939,Kuhn1951} extend the method of Lagrange
multipliers to inequality constraints. For the convex problem
considered here, stationarity, primal and dual feasibility, and
complementary slackness characterise the optimum. The linear
constraints permit the required multipliers to be obtained without
an additional constraint qualification, as discussed in
Section~\ref{sec:theory}.

For (\ref{eq:CLASSOprobdef}), the corresponding Lagrangian is 
\begin{equation}
f_\lambda(\bm{\beta},\bm{\mu})
= \tfrac{1}{2}\lVert\bm{y}-\bm{X}\bm{\beta}(\lambda)\rVert_2^2
+ \lambda \lVert \bm{\beta}(\lambda) \rVert_1 + \bm{\mu}^{\rm
t}(\lambda)
[\bm{A\beta}(\lambda) - \bm{b}],
\label{eq:KKTfdef}
\end{equation}
where $\bm{\mu}(\lambda)$ is a $c$-vector of KKT multipliers and we
have made explicit which variables on the RHS depend on $\lambda$. The multipliers\footnote{In this paper, we adopt
(a slightly modified version of) the notation of \cite{Gaines2018}
whereby for any vector $\bm{v}$ and set of indices $\mathcal{S}$, we
denote by $\bm{v}_\mathcal{S}$ the sub-vector of size $|\mathcal{S}|$
containing the elements of $\bm{v}$ corresponding to the indices in
$\mathcal{S}$. Similarly, for a matrix $\bm{M}$ and another index set
$\mathcal{T}$, we denote by $\bm{M}_{\mathcal{S}\,\mathcal{T}}$ the
sub-matrix containing the rows and columns corresponding to the
indices in $\mathcal{S}$ and $\mathcal{T}$, respectively. We use a
colon, :, when all indices along one of the dimensions are included,
i.e.\ $\bm{M}_{\mathcal{S}\,:}$ contains the rows corresponding to
$\mathcal{S}$ but all of the columns of $\bm{M}$. With some
overloading of the notation, we also denote an element of the vector
$\bm{v}_\mathcal{S}$ or the matrix $\bm{M}_{\mathcal{S}\,\mathcal{T}}$
by $v_\mathcal{S}$ or $M_{\mathcal{S}\,\mathcal{T}}$,
respectively. Further, we denote the complement of the set
$\mathcal{S}$ by $\bar{\mathcal{S}}$.}  $\bm{\mu}_e(\lambda)$
corresponding to equality (=) constraints are free in sign,
whereas the remaining multipliers corresponding to inequality $(\le)$
constraints satisfy $\bm{\mu}_{\bar e}(\lambda)\ge \bm{0}$.  Since the
$L_1$ penalty term in (\ref{eq:KKTfdef}) is convex and
non-differentiable when any of the $\beta$-values is zero, one must
consider the {\it sub}gradient \citep[see, for
example,][]{Rockafellar1970,Bubeck2015} of (\ref{eq:KKTfdef}) with
respect to $\bm{\beta}$, yielding
\begin{equation}
\partial_{\bm{\beta}}f_\lambda(\bm{\beta},\bm{\mu}) = \bm{H\beta}(\lambda)
- \bm{h} + \lambda\bm{g}(\lambda) + \bm{A}^{\rm t}\bm{\mu}(\lambda),
\label{eq:KKTfderiv}
\end{equation}
where $\bm{H}\equiv\bm{X}^{\rm t}\bm{X}$ and
$\bm{h}\equiv\bm{X}^{\rm t}\bm{y}$. Thus $\bm{H}$ is the
Hessian of the quadratic loss, and
$\hat{\bm{\beta}}_{\rm OLS}=\bm{H}^{-1}\bm{h}$ is its
unconstrained minimiser.

The variables $\bm{g}(\lambda)$ in
(\ref{eq:KKTfderiv}) are the subgradients
$\partial_{\bm{\beta}}\lVert \bm{\beta}\rVert_1$ of the $L_1$ penalty
term, with elements
\begin{equation}
\begin{array}{ll}
-1 \le g_p(\lambda) \le 1 & \mbox{if $\beta_p(\lambda) = 0$,} \\
g_p(\lambda) = \mbox{sign}[\beta_p(\lambda)] &
\mbox{otherwise}.
\end{array}
\label{eq:subgrads}
\end{equation}
According to the KKT theorem, the optimal point
$\bm{\beta}(\lambda)$ of (\ref{eq:CLASSOprobdef}) is characterised by
the vanishing of the RHS of (\ref{eq:KKTfderiv}), namely
\begin{equation}
\bm{H\beta}(\lambda)
- \bm{h} + \lambda\bm{g}(\lambda) + \bm{A}^{\rm t}\bm{\mu}(\lambda) = \bm{0},
\label{eq:KKTcond1}
\end{equation}
together with primal feasibility, (\ref{eq:subgrads}), and
dual feasibility and complementary slackness for the inequality
rows. In particular, a slack inequality has zero multiplier,
whereas a binding inequality may have either a positive or a
zero multiplier. The latter case, referred to as weak activity,
is important for determining the outgoing direction.
Differentiating (\ref{eq:KKTcond1}) along a smooth portion of the
path gives
\begin{equation}
\bm{H}\dot{\bm{\beta}}(\lambda) + [\lambda\bm{g}(\lambda)]\dot{} + \bm{A}^{\rm
t}\dot{\bm{\mu}}(\lambda) = \bm{0},
\label{eq:dot}
\end{equation}
where a dot denotes ${\rm d}/{\rm d}\lambda$. In what follows, we will
omit the explicit dependence on $\lambda$ of the variables
$\bm{\beta}$, $\bm{g}$ and $\bm{\mu}$ to avoid notational clutter, but
this dependency should be borne in mind.

\section{Prior path-following approach of Gaines et al.}\label{sec:gaines_prior}

The construction of \citet{Gaines2018} follows the evolution of
the active coefficient set $a$ and the binding constraint set
$z$. The former contains the non-zero coefficients, while the
latter contains all equality rows and the inequality rows that
are saturated. Thus $\bm{\beta}_{\bar a}=\bm{0}$,
$\bm{\mu}_{\bar z}=\bm{0}$, and the elements of $\bm{g}_a$
are fixed by the signs of $\bm{\beta}_a$. The binding
constraints reduce to
\begin{equation}
\bm{A}_{za}\,\bm{\beta}_a = \bm{b}_z.
\label{eq:abcog}
\end{equation}
Restricting (\ref{eq:KKTcond1}) to the active set yields a 
condition relating $\bm{\beta}_a$ and $\bm{\mu}_z$, given by
\begin{equation}
\bm{H}_{aa}\bm{\beta}_a
- \bm{h}_a + \lambda\bm{g}_a + \bm{A}_{za}^{\rm t}\bm{\mu}_z = \bm{0},
\label{eq:KKTcond1a}
\end{equation}
where $\bm{g}_a$ has components $\pm 1$. Conversely, restricting
(\ref{eq:KKTcond1}) to the inactive set yields a condition 
that defines $\bm{g}_{\bar{a}}$, namely
\begin{equation}
\bm{H}_{\bar{a}a}\bm{\beta}_a
- \bm{h}_{\bar{a}} + \lambda\bm{g}_{\bar{a}} + \bm{A}_{z\bar{a}}^{\rm
  t}\bm{\mu}_z = \bm{0}.
\label{eq:KKTcond1abar}
\end{equation}
These equations determine a primal--dual certificate at any
specified $\lambda$. One could therefore solve for such a
certificate separately at each point of a grid,\footnote{A
standard alternative is to solve the original constrained
problem (\ref{eq:CLASSOprobdef}) as a $2p$-dimensional QP at
each requested $\lambda$.} but the affine structure between
changes of state allows the intervening solutions to be obtained
by propagation.

For a path segment along which the active and binding sets remain
unchanged, as do the signs of $\bm{g}_{a}$, \citet{Gaines2018} consider
(\ref{eq:dot}) restricted to the active and binding sets to
obtain a condition relating $\dot{\bm{\beta}}_a$ and $\dot{\bm{\mu}}_z$ given
by
\begin{equation}
\bm{H}_{aa}\dot{\bm{\beta}}_a + \bm{g}_a
+ \bm{A}^{\rm t}_{za}\dot{\bm{\mu}}_z = \bm{0},
\label{eq:dotcond1}
\end{equation}
where $\bm{g}_a$ has components $\pm 1$. When combined with the
derivative with respect to $\lambda$ of the binding constraints
(\ref{eq:abcog}), this gives
\begin{equation}
\left(
\begin{array}{l}
\dot{\bm{\beta}}_{a} \\
\dot{\bm{\mu}}_z
\end{array}
\right)
=
-\left(\begin{array}{ll}
\bm{H}_{aa} & \bm{A}_{za}^{\rm t} \\
 \bm{A}_{za} & \bm{0}
\end{array}\right)^{-1}
\left(
\begin{array}{l}
\bm{g}_{a} \\
\bm{0}
\end{array}
\right).
\label{eq:linearpath}
\end{equation}
Since the RHS of (\ref{eq:linearpath}) is constant, and
recalling also that $\bm{\beta}_{\bar{a}} = \bm{0}$ and
$\bm{\mu}_{\bar{z}} = \bm{0}$, such a path segment is therefore {\it
linear} in the space $(\bm{\beta},\bm{\mu})$. Thus, the full optimal solution path in the space $(\bm{\beta},\bm{\mu})$ is
\emph{piecewise-linear}, with `kinks' potentially occurring only
where the active set $a$ or binding set $z$ changes.
Moreover, along each
affine segment, (\ref{eq:dot}) restricted to the inactive
and binding sets yields
\begin{equation}
(\lambda \bm{g}_{\bar{a}})\dot{} =
-\bm{H}_{\bar{a}a} \dot{\bm{\beta}}_{a} - \bm{A}^{\rm
t}_{z\bar{a}}\dot{\bm{\mu}}_z,
\label{eq:sgevol}
\end{equation}
where the RHS is constant, and so
the quantity $\lambda \bm{g}_{\bar{a}}$ also changes linearly. Similarly,
the inequality residual
$\bm{r}_{\bar{z}} \equiv \bm{A}_{\bar{z}a}\bm{\beta}_a-\bm{b}_{\bar{z}}$
changes linearly, since
\begin{equation}
\dot{\bm{r}}_{\bar{z}} = \bm{A}_{\bar{z}a}\dot{\bm{\beta}}_a.
\label{eq:ineqres}
\end{equation}

\citet{Gaines2018} use (\ref{eq:linearpath})--(\ref{eq:ineqres})
to monitor changes to $a$ and $z$, which can potentially result in
kinks in the solution path, as $\lambda$ is decreased from
$\lambda_{\rm max}$ towards $\lambda = 0$.  In particular, if
$(\bm{\beta}^{(k)},\bm{\mu}^{(k)},\bm{g}^{(k)})$ is the solution at kink $k$, where
$\lambda = \lambda^{(k)}$, then
the next kink $k+1$ is identified by the smallest $\Delta > 0$ that
satisfies at least one of the following conditions for {\it any} index
in the indicated sets:
\begin{subequations}
\begin{eqnarray}
\beta^{(k)}_a - \dot{\beta}_a\Delta & = &  0, \\ 
\lambda^{(k)}g^{(k)}_{\bar{a}} - (\lambda g_{\bar{a}})\dot{}\,\Delta &
= & \pm (\lambda^{(k)}-\Delta),\\
\mu^{(k)}_z - \dot{\mu}_z\Delta & = &  0, \\
r^{(k)}_{\bar{z}} - \dot{r}_{\bar{z}}\Delta & = & 0.
\end{eqnarray}
\end{subequations}

In the absence of constraints this construction recovers the
LASSO-modified LARS method \citep{Efron2004}. Under constraints,
however, two issues require attention. First, an inactive
coefficient may already lie on a subgradient boundary at the
current kink. As discussed by \citet{Gaines2018}, continued
propagation would take its subgradient outside $[-1,1]$
if
\begin{equation}
g^{(k)}_{\bar{a}} (\lambda g_{\bar{a}})\dot{} < 1. 
\label{eq:sglimit}
\end{equation}
Such coordinates are moved into the active set and the direction
is recalculated before propagation. Second, the explicit inverse
in (\ref{eq:linearpath}) requires an appropriate rank condition.
For a positive-definite $\bm{H}_{aa}$, the relevant condition is
full row rank of $\bm{A}_{za}$. Although this condition is
already given by \citet{Gaines2018} and follows from standard
saddle-point theory \citep{Benzi2005}, it is useful to record the
equivalence and nullity explicitly, since they allow the
singular cases to be identified in the examples below.

\begin{proposition}[Invertibility criterion for the Gaines matrix]
\label{prop:gaines_invertibility}
Let
\[
\bm G=
\begin{pmatrix}
\bm H_{aa} & \bm A_{za}^{\rm t}\\
\bm A_{za} & \bm 0
\end{pmatrix},
\]
where $\bm H_{aa}$ is symmetric positive definite. Then the following
are equivalent:
\begin{enumerate}
\item $\bm G$ is invertible;
\item $\bm A_{za}$ has full row rank;
\item the Schur complement
$-\bm A_{za}\bm H_{aa}^{-1}\bm A_{za}^{\rm t}$ is invertible.
\end{enumerate}
Moreover,
$\dim\ker(\bm G)=|z|-\operatorname{rank}(\bm A_{za}).$
\end{proposition}

\begin{proof}
Since $\bm H_{aa}$ is invertible, the block matrix $\bm G$ is
invertible if and only if its Schur complement
$
\bm S=-\bm A_{za}\bm H_{aa}^{-1}\bm A_{za}^{\rm t}
$
is invertible. Because $\bm H_{aa}^{-1}$ is positive definite,
\[
\operatorname{rank}(\bm S)
=\operatorname{rank}(\bm A_{za}\bm H_{aa}^{-1}\bm A_{za}^{\rm t})
=\operatorname{rank}(\bm A_{za}),
\]
and hence $\bm S$ is invertible if and only if
$\operatorname{rank}(\bm A_{za})=|z|$, i.e.\ iff $\bm A_{za}$ has full
row rank.

To compute the nullity, let $\bm{w} \equiv (\bm u\,\,\,\bm v)^{\rm t}$ satisfy
$
\bm G\bm{w}={\bm 0}.
$
Then
\[
\bm H_{aa}\bm u+\bm A_{za}^{\rm t}\bm v=\bm 0,\qquad
\bm A_{za}\bm u=\bm 0.
\]
Eliminating $\bm u=-\bm H_{aa}^{-1}\bm A_{za}^{\rm t}\bm v$ yields
\[
\bm A_{za}\bm H_{aa}^{-1}\bm A_{za}^{\rm t}\bm v=\bm 0.
\]
Thus $\bm v$ lies in the nullspace of
$\bm A_{za}\bm H_{aa}^{-1}\bm A_{za}^{\rm t}$, whose nullity is
$|z|-\operatorname{rank}(\bm A_{za})$. Since $\bm u$ is then uniquely
determined by $\bm v$, the same nullity holds for $\bm G$.
\end{proof}

\begin{corollary}[Two diagnostic causes of Gaines singularity]
\label{cor:gaines_failure_modes}
Under the assumptions of Proposition~\ref{prop:gaines_invertibility},
the Gaines matrix $\bm G$ is singular if and only if the rows of
$\bm A_{za}$ are linearly dependent. Two cases are worth
distinguishing diagnostically:
\begin{enumerate}
\item[\textit{(i)}] \emph{dimension obstruction}: $|z|>|a|$, in which
row dependence is automatic and no property of the data can avoid it;
\item[\textit{(ii)}] \emph{genuine dependence}: $|z|\le|a|$, yet the
rows of $\bm A_{za}$ are still linearly dependent.
\end{enumerate}
Conversely, if $|z|\le |a|$ and $\bm A_{za}$ has full row rank, then
$\bm G$ is invertible.
\end{corollary}

\begin{proof}

The equivalence follows from
Proposition~\ref{prop:gaines_invertibility}(2). If $|z|>|a|$,
the matrix $\bm{A}_{za}$ has more rows than columns and cannot
have full row rank. The complementary case is dependence when
$|z|\le|a|$. The converse follows from the same proposition.
\end{proof}

\begin{remark}
The two diagnostic cases are mutually exclusive. The first can
be detected from the set cardinalities, whereas the second
requires a rank calculation. Both arise from row dependence.
For example, with non-negativity constraints
$\bm{C}=-\bm{I}_p$ and $\bm{d}=\bm{0}_p$, each binding
row is supported on a zero coefficient, so its restriction to
the active columns vanishes. Hence $\bm{A}_{za}=\bm{0}$
and the nullity is $|z|$ whenever any coefficient remains at
its bound. The diabetes examples in
Section~\ref{sec:gaines_comparison} illustrate this situation.
\end{remark}

\begin{remark}
Boundary-set variables do not enter the rank condition in
Proposition~\ref{prop:gaines_invertibility} directly. Their role in
the Gaines algorithm is different: they require a separate workaround
in which such variables are moved into the active set and the linear
system is recomputed. This may alter $\bm A_{za}$ and hence can
indirectly move the method into or out of the singular regime, but it
is not itself a singularity mechanism.
\end{remark}

\section{The CLASSO path-following algorithm}\label{sec:new_approach}

We now obtain the local direction by solving a QP whose
constraints distinguish the possible behaviour of boundary
coefficients and binding inequalities. The primal solution
$\bm{\beta}(\lambda)$ is propagated along the path, while the
dual variables provide a local certificate and determine the
admissible step. Suppose that
$(\bm{\beta}^{(k)},\bm{\mu}^{(k)},\bm{g}^{(k)})$ is such a
certificate at $\lambda^{(k)}$. On an outgoing affine segment,
a decrease $\Delta>0$ in $\lambda$ gives
\begin{subequations}
\label{eq:deltas}
\begin{eqnarray}
\lambda & = & \lambda^{(k)}-\Delta,\label{eq:dlambda}\\
\bm{\beta} & = & \bm{\beta}^{(k)} - \dot{\bm{\beta}}\Delta,\label{eq:dbeta}\\
\bm{\mu} & = & \bm{\mu}^{(k)} - \dot{\bm{\mu}}\Delta,\label{eq:dmu}\\
\lambda\bm{g} & = & \lambda^{(k)}\bm{g}^{(k)} - (\lambda\bm{g})\dot{}\,\Delta,\label{eq:dlg}
\end{eqnarray}
\end{subequations}
where $\dot{\bm{\beta}}$, $\dot{\bm{\mu}}$ and $(\lambda\bm{g})\dot{}$
are constant vectors.

Define the \emph{augmented active set} $\alpha=a\cup b$,
where $b\subset\bar a$ contains the zero coefficients whose
subgradients lie on the boundary of $[-1,1]$. Thus
$\beta_j=0$ and $|g_j|=1$ for $j\in b$. These coefficients
may begin to move on the outgoing segment, although reaching
the subgradient boundary does not require them to do so.
Consequently, for sufficiently small $\Delta$ we may impose
$\dot{\bm{\beta}}_{\bar\alpha}=\bm{0}$, but must allow the
direction calculation to determine
$\dot{\bm{\beta}}_b$.

For the components in $b$, we first ensure consistency with the
subgradient rule (\ref{eq:subgrads}) for the $L_1$ norm. If $j\in b$, then
$\beta_j^{(k)}=0$ and $g_j^{(k)}=\pm1$, and (\ref{eq:dbeta}) gives
$\beta_j=-\dot\beta_j\Delta$. Thus, 
if $g_j^{(k)}=+1$ $(-1)$, then $\beta_j$ may move only into $\beta_j\ge 0$ $(\le 0)$, and hence $\dot\beta_j\le0$ $(\ge 0)$.
Thus, for all $j\in b$, the allowable direction is characterised by
$g_j^{(k)}\dot\beta_j\le0$, or, in vector form,
\begin{equation}
\bm g_b^{(k)}\odot \dot{\bm\beta}_b \le \bm 0,
\label{eq:sign_compatibility}
\end{equation}
where $\odot$ denotes the elementwise Hadamard product.
Equivalently, if $\bm{D}_b$ denotes the
diagonal matrix whose diagonal entries are the components of
$\bm{g}^{(k)}_b$, each of which is equal to $\pm 1$, then (\ref{eq:sign_compatibility})
may be
written as $\bm{D}_b\dot{\bm{\beta}}_b\le \bm{0}$.

The same local geometry expressed by (\ref{eq:sign_compatibility}) on the
primal side may also be expressed on the dual (subgradient) side. For
$j\in b$, write the local evolution (\ref{eq:dlg}) in the form
\begin{equation}
\lambda g_j=\lambda^{(k)}g_j^{(k)}-q_j\Delta.
\label{eq:gj_local}
\end{equation}
Since $\lambda=\lambda^{(k)}-\Delta$, the requirement that
$g_j\in[-1,1]$ along the segment implies that $q_j\ge1$ $(\le -1)$ when
$g_j^{(k)}=+1$ $(-1)$. Hence, one may write
\begin{equation}
q_j=g_j^{(k)}(1+\rho_j),
\qquad
\rho_j\ge0.
\label{eq:qj_rho}
\end{equation}
Collecting these relations componentwise gives
\begin{equation}
\lambda\bm{g}_b
=
\lambda^{(k)}\bm{g}^{(k)}_b
-
\left(\bm{g}^{(k)}_b+\bm{D}_b\bm{\rho}_b\right)\Delta,
\qquad
\bm{\rho}_b\ge\bm{0}.
\label{eq:lgb}
\end{equation}
The case $\rho_j=0$ corresponds to the subgradient remaining at the
boundary value $\pm1$, whereas $\rho_j>0$ corresponds to it moving
into the interior of the interval $[-1,1]$. Thus
(\ref{eq:sign_compatibility}) and (\ref{eq:lgb}) are complementary
primal and dual descriptions of the same local behaviour on the set
$b$. It is also worth noting that, for the original active set $a$, the
subgradient remains fixed (with elements $\pm 1$) along the current
segment, and so
\begin{equation}
\lambda\bm{g}_a
=
\lambda^{(k)}\bm{g}^{(k)}_a-\bm{g}^{(k)}_a\Delta.
\label{eq:lga}
\end{equation}

The linear constraints require a further partition of the
binding set. Write $z=e\cup u$, where $e$ contains the
equality rows and $u$ the binding inequalities. The strongly
active set $s\subseteq u$ consists of rows with
$\bm{\mu}^{(k)}_s>\bm{0}$, while the weakly active set
$w=u\setminus s$ has $\bm{\mu}^{(k)}_w=\bm{0}$.
Equality rows and strongly active inequalities remain binding
on a sufficiently short outgoing segment. A weakly active
inequality, however, may become slack and need only remain
feasible.

The rates $\dot{\bm{\mu}}_e$ and
$\dot{\bm{\mu}}_s$ are therefore unrestricted in sign.
For a weakly active row,
$\bm{\mu}_w=-\Delta\dot{\bm{\mu}}_w\ge\bm{0}$
requires $\dot{\bm{\mu}}_w\le\bm{0}$. The remaining
inequalities are slack at the current point and have zero
multipliers. Thus $z=e\cup s\cup w$, with $\bar z$
denoting the slack inequalities.

Substituting (\ref{eq:dbeta}) into the constraints in
(\ref{eq:CLASSOprobdef}), and noting that
$\bm{A}_{e:}\bm{\beta}^{(k)}=\bm{b}_e$ and
$\bm{A}_{u:}\bm{\beta}^{(k)}=\bm{b}_u$ at the kink, one obtains the
first-order feasibility conditions
\begin{equation}
\bm{A}_{e\alpha}\dot{\bm{\beta}}_\alpha = \bm{0},
\qquad
\bm{A}_{s\alpha}\dot{\bm{\beta}}_\alpha = \bm{0},
\qquad
-\bm{A}_{w\alpha}\dot{\bm{\beta}}_\alpha \le \bm{0}.
\label{eq:dotconst}
\end{equation}

In particular, one cannot impose all binding inequalities as
equalities in the directional problem: the rows in $w$ must
retain their inequality character so that they can become slack.

We now derive the local path direction. Restricting the KKT
optimality condition (\ref{eq:KKTcond1}) to the set $\alpha$ gives
\begin{equation}
\bm{H}_{\alpha\alpha}\bm{\beta}_\alpha
-\bm{h}_\alpha
+\lambda\bm{g}_\alpha
+\bm{A}^{\rm t}_{e\alpha}\bm{\mu}_e
+\bm{A}^{\rm t}_{s\alpha}\bm{\mu}_s
+\bm{A}^{\rm t}_{w\alpha}\bm{\mu}_w
=
\bm{0}.
\label{eq:kkt_alpha_new}
\end{equation}
Substituting (\ref{eq:dbeta})--(\ref{eq:dmu}) together with
(\ref{eq:lgb})--(\ref{eq:lga}) into (\ref{eq:kkt_alpha_new}),
cancelling the zeroth-order terms using
(\ref{eq:kkt_alpha_new}) at the kink, and dividing by $-\Delta$,
yields
\begin{equation}
\bm{H}_{\alpha\alpha}\dot{\bm{\beta}}_\alpha
+\bm{g}^{(k)}_\alpha
+\bm{D}_b^{\rm t}\bm{\rho}_b
+\bm{A}^{\rm t}_{e\alpha}\dot{\bm{\mu}}_e
+\bm{A}^{\rm t}_{s\alpha}\dot{\bm{\mu}}_s
+\bm{A}^{\rm t}_{w\alpha}\dot{\bm{\mu}}_w
=
\bm{0}.
\label{eq:dotalpha}
\end{equation}
On comparing (\ref{eq:dotalpha}) with the KKT optimality conditions,
one sees that it is precisely the stationarity condition for the
$|\alpha|$-dimensional QP
\begin{equation}
\min_{\dot{\bm{\beta}}_\alpha}
\tfrac{1}{2}\dot{\bm{\beta}}^{\rm t}_\alpha\bm{H}_{\alpha\alpha}
\dot{\bm{\beta}}_\alpha
+\left(\bm{g}^{(k)}_\alpha\right)^{\rm t}\dot{\bm{\beta}}_\alpha
\label{eq:dotprobdef}
\end{equation}
subject to
\begin{equation}
\bm{A}_{e\alpha}\dot{\bm{\beta}}_\alpha=\bm{0},
\quad
\bm{A}_{s\alpha}\dot{\bm{\beta}}_\alpha=\bm{0},
\quad
-\bm{A}_{w\alpha}\dot{\bm{\beta}}_\alpha\le \bm{0},
\quad
\bm{g}^{(k)}_b\odot \dot{\bm{\beta}}_b \le \bm{0}.
\label{eq:dotprobdef_constraints}
\end{equation}
The corresponding KKT multipliers are $\dot{\bm{\mu}}_e$,
$\dot{\bm{\mu}}_s$, $-\dot{\bm{\mu}}_w$ and $\bm{\rho}_b$,
respectively, since the weakly active constraints are written in the
form $-\bm{A}_{w\alpha}\dot{\bm{\beta}}_\alpha\le \bm{0}$ (one could,
more conventionally, introduce the multipliers $\bm{\xi}_w \equiv
-\dot{\bm{\mu}}_w \geq \bm{0}$ for the weakly active constraints, but
we prefer the existing notation). Since $\bm{H}_{\alpha\alpha}$ is
positive definite, the minimiser $\dot{\bm{\beta}}_\alpha$ is
unique. In particular, this formulation does not require inversion of
the matrix in (\ref{eq:linearpath}), and so does not rely on the rows
of $\bm{A}_{z\alpha}$ being linearly independent. The remaining
non-trivial quantities are $(\lambda\bm{g}_{\bar{\alpha}})\dot{}$,
which may be determined by considering the equivalent of the
first-order stationarity condition (\ref{eq:dotalpha}), but restricted
instead to the set $\bar{\alpha}$. This yields
\begin{equation}
(\lambda\bm{g}_{\bar{\alpha}})\dot{} =
-\bm{H}_{\bar{\alpha}\alpha}\dot{\bm{\beta}}_\alpha
-\bm{A}^{\rm t}_{e\bar{\alpha}}\dot{\bm{\mu}}_e
-\bm{A}^{\rm t}_{s\bar{\alpha}}\dot{\bm{\mu}}_s
-\bm{A}^{\rm t}_{w\bar{\alpha}}\dot{\bm{\mu}}_w.
\label{eq:sgevol_new}
\end{equation}

Uniqueness of the primal direction does not imply uniqueness of
the directional multipliers. When the active directional rows
are dependent, different multiplier solutions may give
different rates for $\lambda\bm{g}_{\bar\alpha}$ in
(\ref{eq:sgevol_new}), even though
$\dot{\bm{\beta}}_\alpha$ is unchanged. These rates affect
the subsequent certificate events and hence require a specified
selection rule. We follow the certificate itself. At a point
$(\bm\beta,\lambda)$ the valid multipliers form the set
\begin{equation}
\mathcal{M}(\bm\beta,\lambda)=\left\{\bm\mu:\;
\begin{array}{l}
\lambda g_j=\lambda\,\mathrm{sign}\,\beta_j\ (j\in a),\\
|\lambda g_j|\le\lambda\ (j\in\bar a),\\
\mu_i\ge0,\ \ \mu_i(\bm{A}_{i:}\bm\beta-b_i)=0\ \ (i\in\mathcal{I})
\end{array}\right\},
\label{eq:valid_multipliers}
\end{equation}
where $\lambda\bm g=\bm h-\bm H\bm\beta-\bm A^{\rm t}\bm\mu$ is
determined by $\bm\mu$ through stationarity and $\mathcal{I}$ indexes
the inequality rows. It is a polyhedron in $\bm\mu$ alone, nonempty
at a solution, and so has a unique point of least Euclidean norm,
\begin{equation}
\bm\mu^\star(\bm\beta,\lambda)=\mathop{\rm argmin}\bigl\{\lVert\bm\mu\rVert:
\bm\mu\in\mathcal{M}(\bm\beta,\lambda)\bigr\},
\label{eq:canonical_certificate}
\end{equation}
which we call the \emph{canonical certificate}; $\bm g$ is recovered
from stationarity and does not enter the norm.

Fix the outgoing direction $\dot{\bm\beta}$ and the outgoing pattern of
signs, binding rows and held coordinates. Along the segment
(\ref{eq:deltas}), the set
$\mathcal{M}(\bm\beta^{(k)}-\Delta\dot{\bm\beta},\lambda^{(k)}-\Delta)$
restricted to that pattern is
$\{\bm\mu:\bm r_i^{\rm t}\bm\mu\le c_i+\Delta\,\dot c_i\}$, with fixed
normals $\bm r_i$ and affinely moving bounds, the equalities being
written as pairs of inequalities. Write
$\mathcal{T}=\{i:\bm r_i^{\rm t}\bm\mu^\star=c_i\}$ for the rows tight at
$\bm\mu^\star=\bm\mu^\star(\bm\beta^{(k)},\lambda^{(k)})$, and parametrise
the certificate as $\bm\mu^\star+\Delta\bm w$, so that
$\bm w=-\dot{\bm\mu}$ in (\ref{eq:dmu}). The rates that keep it
feasible to first order form
$\mathcal{W}=\{\bm w:\bm r_i^{\rm t}\bm w\le\dot c_i,\ i\in\mathcal{T}\}$,
and the \emph{lexicographic rate} is
\begin{equation}
\bm w^\star=\mathop{\rm argmin}\Bigl\{\lVert\bm w\rVert:\bm w\in\mathcal{W},\;
\bm\mu^{\star\rm t}\bm w=\min_{\bm w'\in\mathcal{W}}\bm\mu^{\star\rm t}\bm w'\Bigr\}:
\label{eq:lex_rate}
\end{equation}
first the fastest decrease of $\tfrac12\lVert\bm\mu\rVert^2$, whose
derivative is $\bm\mu^{\star\rm t}\bm w$, then the least norm. The
inner minimum is finite because $-\bm\mu^\star$ lies in the cone of
the tight normals, and $\bm w^\star$ is unique by strict convexity.
Let $\mathcal{T}^\star=\{i\in\mathcal{T}:\bm r_i^{\rm t}\bm w^\star=\dot c_i\}$
be the rows that the rate keeps tight. The propagated certificate
remains the minimum-norm point while it remains feasible and
\begin{equation}
-(\bm\mu^\star+\Delta\bm w^\star)\in
\mathrm{cone}\{\bm r_i:i\in\mathcal{T}^\star\},
\label{eq:cert_horizon}
\end{equation}
which is the optimality condition of (\ref{eq:canonical_certificate})
on those rows. The \emph{certificate horizon} is
\begin{equation}
\Delta_{\rm cert}=\sup\bigl\{\Delta\ge0:
-(\bm\mu^\star+\Delta\bm w^\star)\in
\mathrm{cone}\{\bm r_i:i\in\mathcal{T}^\star\}\bigr\}\in(0,\infty];
\end{equation}
it is infinite, for example, when $\bm\mu^\star$ and $\bm w^\star$
both vanish, and then no \textbf{E6} is scheduled, the step being
bounded by the other events and by $\lambda^{(k)}$. Feasibility of
$\bm\mu^\star+\Delta\bm w^\star$ is the subject of events
\textbf{E2}, \textbf{E4} and \textbf{E5} below, and a finite horizon
defines \textbf{E6}. Following this certificate, rather than propagating a
rate of our own choosing and recovering a new certificate at the next
kink, is what makes the traversal terminate
(Section~\ref{sec:termination}).

The quantities (\ref{eq:deltas}) are propagated along the current
linear segment and the next kink $k+1$ is identified by the smallest
$\Delta > 0$ for which any of the following conditions is met for any
index in the indicated set:
\begin{subequations}
\label{eq:kink_events}
\begin{align}
\textbf{E1}\,(\emph{active leaving}): \quad &
\beta_a = 0,
\label{eq:event_E1}\\
\textbf{E2}\, (\emph{subgradient saturation}): \quad &
g_{\bar{\alpha}} = \pm 1,
\label{eq:event_E2}\\
\textbf{E3}\, (\emph{inequality binding}): \quad &
\bm{A}_{\bar{z}\,:}\,\bm{\beta} = \bm{b}_{\bar{z}},
\label{eq:event_E3}\\
\textbf{E4}\, (\emph{multiplier release}): \quad &
\mu_s = 0,
\label{eq:event_E4}\\
\textbf{E5}\, (\emph{held-coordinate sign crossing}): \quad &
(\lambda g_j) = -s_j\,\lambda \quad (j\in b \mbox{ held}),
\label{eq:event_E5}\\
\textbf{E6}\, (\emph{certificate face change}): \quad &
\mbox{horizon of the certificate}.
\label{eq:event_E6}
\end{align}
\end{subequations}
Events \textbf{E1}--\textbf{E4} correspond, respectively, to
a coefficient leaving the support, an inactive subgradient
reaching its boundary, a slack inequality becoming binding,
and a strongly active multiplier reaching zero. A weakly
active row with
$\bm{A}_{w\alpha}\dot{\bm{\beta}}_\alpha>\bm{0}$
becomes slack immediately and requires no subsequent release
event.
Event \textbf{E5} applies to a boundary coordinate $j$ held
at zero by the directional QP. If its initial sign is
$s_j=\pm1$ and the rate of $\lambda g_j$ is $r_j$, then
$\lambda g_j=s_j\lambda^{(k)}-\Delta r_j$. Apart from a
persistent equality, the equation for the original boundary
has only the root $\Delta=0$, whereas the opposite boundary
is reached at
$\Delta=2s_j\lambda^{(k)}/(s_j+r_j)$ whenever this lies
strictly between zero and $\lambda^{(k)}$. Both margins
must therefore be monitored for a held coordinate.
Event \textbf{E6} occurs at the certificate horizon, the end of
the step over which (\ref{eq:cert_horizon}) holds, so that the
lexicographic rate keeps the canonical certificate on its current
face. It
has no primal counterpart: the certificate's rate changes, but $\dot{\bm\beta}$ need not. More generally, a
point at which no coefficient leaves the support and no
inequality becomes binding (neither \textbf{E1} nor \textbf{E3}),
and across which the primal direction is unchanged, is not a kink
of $\bm\beta(\lambda)$ and is not recorded.

The \textbf{E5} crossing is a consequence of the augmented-set
representation. An ordinary LARS implementation treats the
corresponding crossing as an entry, as illustrated by the
comparison with \texttt{lars\_path} in
Section~\ref{sec:benchmarking}. In the parametric-QP
description of \citet{Schmelzer2026}, \textbf{E2} and
\textbf{E5} are the two margins of a single release test.
On our stress suite, \textbf{E5} occurs $11$ times in
$6$ of the $39$ cases, with each occurrence corresponding
to a change of primal direction. These coordinates are held
by their subgradient conditions alone. A coordinate with
a one-sided sign bound cannot enter on the forbidden side;
the redundant subgradient--bound-multiplier decomposition
is treated separately in Section~\ref{sec:one_sided}.

It is worth noting that this path-following approach is easily adapted
for travelling in the direction of increasing $\lambda$ along the
optimal path by making the replacement $-\Delta \to \Delta$ in
(\ref{eq:deltas}) (again with $\Delta > 0$) and changing the signs of
the constraints in
(\ref{eq:dotprobdef_constraints}). Thus, the algorithm can be
tailored to travel in either direction along the optimal path, as the
user prefers. We will assume henceforth, however, that path-following
proceeds in the usual direction of decreasing $\lambda$.

For the unconstrained LASSO, the directional construction
recovers the LASSO-modified LARS direction
\citep{Efron2004}. In the usual single-entry case,
$z=\varnothing$ and $\alpha$ contains the current support
together with the entering coordinate. When the sign-cone
constraint is inactive at the outgoing solution, stationarity
reduces to
$\dot{\bm{\beta}}_\alpha
=-\bm{H}_{\alpha\alpha}^{-1}\bm{g}_\alpha$,
with the leaving event \textbf{E1} providing the usual LASSO
modification. The numerical correspondence is examined in
Section~\ref{sec:benchmarking}.

At the origin the active set is empty. The directional QP over
the full cone (\ref{eq:dotprobdef})--(\ref{eq:dotprobdef_constraints})
still has a unique minimiser, which gives the direction in which
coefficients enter, including several at once; no separate
construction is needed to leave the origin. The initialised run
of Section~\ref{sec:termination} starts in exactly this way.

\subsection{Initialisation}\label{sec:init}

Initialisation requires a primal--dual certificate
$(\bm{\beta}^{(0)},\bm{\mu}^{(0)}_u,
\bm{\mu}^{(0)}_e,\bm{g}^{(0)})$ at the largest value
$\lambda^{(0)}$ from which the path is to be followed.
The quadratic loss must be included when selecting this
point if the minimum-$L_1$ feasible set is not a singleton.

Following \citet{Gaines2018}, one first notes that as
$\lambda\to\infty$ the original problem
(\ref{eq:CLASSOprobdef}) is dominated by the $L_1$ penalty term, and
so approaches
\begin{equation}
\min_{\bm{\beta}}\lVert \bm{\beta} \rVert_1
\qquad\mbox{subject to}\qquad
\bm{A\beta}\preceq \bm{b}.
\label{eq:largelambdaprob_new}
\end{equation}
As before, this may be written as a standard $2p$-dimensional linear
program by decomposing
$\bm{\beta}=\bm{\beta}_+-\bm{\beta}_-$ with
$\bm{\beta}_+\ge \bm{0}$ and $\bm{\beta}_-\ge \bm{0}$, so that
$\lVert\bm{\beta}\rVert_1=\bm{e}^{\rm t}(\bm{\beta}_++\bm{\beta}_-)$.

The LP determines the feasible points with smallest $L_1$
norm. If several such points exist, the quadratic loss
selects their large-$\lambda$ limit. We therefore use a
lexicographic initialisation: first minimise the $L_1$ norm,
and then minimise the loss over the resulting solution set.
The second stage is performed in all cases; when the LP
minimiser is unique, it returns that same point.

Accordingly, let $\mathcal{S}^\ast$ denote the solution set of
(\ref{eq:largelambdaprob_new}). The initial primal point
$\bm{\beta}^{(0)}$ is then defined as the solution of
\begin{equation}
\min_{\bm{\beta}}
\;\tfrac{1}{2}\lVert \bm{y}-\bm{X}\bm{\beta}\rVert_2^2
\qquad\mbox{subject to}\qquad
\bm{\beta}\in \mathcal{S}^\ast.
\label{eq:lexicographic_init}
\end{equation}
In practice, $\mathcal{S}^\ast$ is represented by the
first-stage constraints and its optimal objective value,
using the same split variables as the LP. Strict convexity
of the loss in $\bm{\beta}$ ensures a unique initial
coefficient vector, irrespective of any redundancy in that
representation.

Having obtained $\bm{\beta}^{(0)}$, one next determines the smallest
value of $\lambda$ for which $\bm{\beta}^{(0)}$ satisfies the KKT
conditions of the full constrained LASSO problem. Let $a$ denote the
support of $\bm{\beta}^{(0)}$, let $u$ denote the active inequality
constraints at $\bm{\beta}^{(0)}$, and let $e$ denote the equality
constraints. Inactive inequalities do not appear below because their
multipliers vanish by complementary slackness. Restricting the KKT
stationarity condition (\ref{eq:KKTcond1}) to the active and inactive
coefficient sets then gives
\begin{eqnarray}
\bm{h}_a-\bm{H}_{aa}\bm{\beta}^{(0)}_a
-\bm{A}_{ua}^{\rm t}\bm{\mu}_u
-\bm{A}_{ea}^{\rm t}\bm{\mu}_e\!\!\!
&=&\!\!\!
\lambda\,\mathrm{sign}(\bm{\beta}^{(0)}_a),\phantom{AAAa}
\label{eq:initcert_active_derivation}
\\
-\lambda \le
\bm{h}_{\bar a}-\bm{H}_{\bar a a}\bm{\beta}^{(0)}_a
-\bm{A}_{u\bar a}^{\rm t}\bm{\mu}_u
-\bm{A}_{e\bar a}^{\rm t}\bm{\mu}_e
\!\!\!&\le&\!\!\!
\lambda,
\label{eq:initcert_inactive_derivation}
\end{eqnarray}
together with the dual-feasibility condition
\begin{equation}
\bm{\mu}_u\ge \bm{0}.
\label{eq:initcert_nonneg_derivation}
\end{equation}
Equation (\ref{eq:initcert_active_derivation}) enforces the KKT
stationarity condition on the active support, where the subgradients
are fixed by the signs of the non-zero coefficients. Equation
(\ref{eq:initcert_inactive_derivation}) enforces subgradient
consistency on the inactive set, namely that
$-1\le g_{\bar a}\le 1$. Thus the smallest admissible initial value
$\lambda^{(0)}$ is obtained by solving the initial certification
problem
\begin{equation}
\min_{\bm{\mu}_u,\bm{\mu}_e,\lambda}\;\lambda
\label{eq:initcert_obj}
\end{equation}
subject to
(\ref{eq:initcert_active_derivation})--(\ref{eq:initcert_nonneg_derivation}).
This initial certification problem determines the smallest $\lambda$
at which the lexicographically selected point is optimal for the full
constrained LASSO. It is distinct from the dual of the minimum-$L_1$
LP.

At this value $\lambda^{(0)}$, one then chooses a canonical dual
certificate $(\bm{\mu}^{(0)}_u,\bm{\mu}^{(0)}_e)$. 
For reproducibility, we minimise the multiplier norm over
the certificates at this value of $\lambda$, whether or not
the certificate is already unique:
\begin{equation}
\tfrac{1}{2}\lVert\bm{\mu}_u\rVert_2^2+
\tfrac{1}{2}\lVert\bm{\mu}_e\rVert_2^2
\label{eq:init_min_norm_dual}
\end{equation}
subject to
(\ref{eq:initcert_active_derivation})--(\ref{eq:initcert_nonneg_derivation})
with $\lambda=\lambda^{(0)}$. If the initial certification problem has
a unique solution, this second step simply returns it; otherwise it
provides a consistent and reproducible tie-break.

Provided $\lambda^{(0)}>0$, the corresponding inactive subgradients
follow from (\ref{eq:initcert_inactive_derivation}) as
\begin{equation}
\bm{g}^{(0)}_{\bar a}
=
\frac{1}{\lambda^{(0)}}
\left(
\bm{h}_{\bar a}
-\bm{H}_{\bar a a}\bm{\beta}^{(0)}_a
-\bm{A}_{u\bar a}^{\rm t}\bm{\mu}^{(0)}_u
-\bm{A}_{e\bar a}^{\rm t}\bm{\mu}^{(0)}_e
\right).
\label{eq:init_subgrad_recovery}
\end{equation}
The initial boundary set is then
\begin{equation}
b=\{j\in \bar a: |g_j^{(0)}|=1\},
\label{eq:init_bset}
\end{equation}
from which the augmented active set $\alpha=a\cup b$ may be
constructed and the path-following algorithm started. If
$\lambda^{(0)}=0$, then the path is trivial and no path-following step
is required.

If the resulting initial point satisfies $a=\varnothing$, the path
leaves it along the minimiser of the full-cone directional QP, as at
any other point.

\subsection{The CLASSO algorithm in pseudocode}\label{sec:algo_summary}

The ingredients introduced in this section
(the directional QP \eqref{eq:dotprobdef}--\eqref{eq:dotprobdef_constraints},
the canonical certificate and its lexicographic rate, the
events \textbf{E1}--\textbf{E6} of
\eqref{eq:kink_events}, and the initialisation of
Section~\ref{sec:init}) compose the CLASSO algorithm, summarised
in Algorithm~\ref{alg:classo}.

\begin{algorithm}
\caption{CLASSO: constrained-LASSO path-following}\label{alg:classo}
\begin{algorithmic}[1]
\State \textbf{Input:} $\bm X,\bm y$; constraints $(\bm A,\bm b)$;
$0\le\lambda_{\min}$; numerical tolerances and iteration limit;
optional ridge $\epsilon$.
\State \textbf{Output:} recorded points with their certificate flags,
the associated states, a traversal status and a certificate status.
\State Form $\bm H,\bm h$, including ridge augmentation if required.
\State Initialise $(\lambda,\bm\beta)$ by Section~\ref{sec:init}, with
its canonical certificate $(\bm\mu,\bm g)$ and state
$S=(a,b,e,s,w)$; record the point.
\While{$\lambda>\lambda_{\min}$ and the iteration limit is not reached}
\State Solve and refine the directional QP
\eqref{eq:dotprobdef}--\eqref{eq:dotprobdef_constraints} for
the primal rate $\dot{\bm\beta}$ and outgoing state $S_{\rm out}$
(Section~\ref{sec:release_refinement}).
\If{no \textbf{E1} or \textbf{E3} occurred at the last recorded point
and $\dot{\bm\beta}$ is the direction that arrived there}
\State Remove that point, which is not a kink.
\EndIf
\State Take the lexicographic rate $\dot{\bm\mu}$ of the canonical
certificate, and $\bm q=(\lambda\bm g)\dot{}$ (Section~\ref{sec:follow}).
\State $(\lambda',\bm\beta',\bm\mu',\bm g',S',\Delta,c)
\gets \text{\textsc{Advance}}
(\lambda,\bm\beta,\bm\mu,\bm g,S_{\rm out},
\dot{\bm\beta},\dot{\bm\mu},\bm q,\lambda_{\min})$,
where $c$ is the acceptance flag of the canonical certificate
at~$\lambda'$.
\If{$\Delta$ is below the minimum step}
\State Set the traversal status to \emph{step exhausted} and exit the loop.
\EndIf
\If{$c$ records a failure}
\State Increment the certificate-failure count;
the traversal is not interrupted.
\EndIf
\State Record the new point with $c$; replace
$(\lambda,\bm\beta,\bm\mu,\bm g,S)$ by
$(\lambda',\bm\beta',\bm\mu',\bm g',S')$.
\EndWhile
\If{the loop body set no traversal status}
\State Set it to \emph{reached $\lambda_{\min}$} if
$\lambda\le\lambda_{\min}$, and to \emph{iteration limit} otherwise.
\EndIf
\If{the certificate-failure count is nonzero}
\State Mark the path untrustworthy and warn.
\EndIf
\State \Return the recorded points, their flags, and both statuses.
\end{algorithmic}
\end{algorithm}

For \textsc{Advance}, let $\Delta_{\rm event}$ be the
smallest strictly positive admissible root of
\textbf{E1}--\textbf{E6}, with $\Delta_{\rm event}=+\infty$ if
there is no such root. Set
$\Delta=\min\{\Delta_{\rm event},\lambda-\lambda_{\min}\}$,
propagate the coefficients by (\ref{eq:deltas}) and the
certificate along its lexicographic rate, and update the state
using the event conditions. At the new point, compute its
canonical certificate as in Section~\ref{sec:exact_certificate};
in exact arithmetic it is the propagated one. Return the point
reached, its updated state, the step $\Delta$ taken, and the
certificate with its acceptance flag. A point whose certificate
is not accepted is still returned, flagged: certification
governs whether the output may be believed, not whether the
traversal proceeds. If the step reaches $\lambda_{\min}$, the
endpoint is certified and recorded in the same way, and the
loop's own condition then ends the traversal. At $\lambda=0$,
stationarity is checked with $\lambda\bm g=\bm0$, without
dividing by $\lambda$. If $\lambda^{(0)}=0$, the initialisation
already returns the terminal solution.

The pseudocode summarises the mathematical operations; the
numerical acceptance tests and fallbacks are described in
Section~\ref{sec:implementation} and
Appendix~\ref{sec:safeguards}. In particular, the implementation
reports two statuses separately. The traversal status records
whether a run reaches $\lambda_{\min}$ or stops at the
iteration limit or on a vanishing step. The certificate status
records how many recorded points failed certification; if any
did, the path is marked untrustworthy and a warning is issued,
and the output is then a recorded prefix with certificate flags
rather than a certified path.

For a fixed state, the principal dense factorisation costs
are $O(|\alpha|^3)$ for the directional calculation and
$O(|z_{\mathrm{act}}|^2|\alpha|)$ for multiplier recovery,
where $|z_{\mathrm{act}}|$ is the number of active constraint
rows. The total cost depends on the number of iterations
performed, including any refinement or recovery steps.
Corollary~\ref{cor:complexity} bounds the number of affine
pieces of the primal path, and Section~\ref{sec:termination}
bounds the number of iterations of the algorithm in exact
arithmetic; both bounds are exponential in the worst case.
Empirically, the recorded point count
is approximately proportional to $p$ for the configurations
examined in Section~\ref{sec:scaling}.

\section{Well-posedness and path finiteness}\label{sec:theory}

We first establish that the directional QP is well posed and
gives the local direction of the primal path. We then prove
that this path is continuous and piecewise affine with
finitely many kinks. Throughout, $\bm{H}\succ\bm{0}$,
either directly or after the ridge augmentation of
Section~\ref{sec:problem}, and the feasible set is assumed
nonempty. No constraint qualification, strict complementarity
or uniqueness of multipliers is required.

These statements concern the mathematical path and its local
continuation. Whether a particular event and certificate
selection rule traverses the path in finitely many steps is
a separate question, considered below.

\begin{proposition}[Well-posedness of the CLASSO directional QP]
\label{prop:directional_qp_wellposed}
Fix any admissible state $(a,b,e,s,w)$ and let $\alpha=a\cup b$.
The directional problem
(\ref{eq:dotprobdef})--(\ref{eq:dotprobdef_constraints}) has
nonempty feasible set and a unique primal minimiser. In particular,
no rank condition on $\bm{A}_{z\alpha}$ is required.
\end{proposition}

\begin{proof}
The feasible set is nonempty, since
$\dot{\bm{\beta}}_\alpha=\bm{0}$ satisfies all of the constraints in
(\ref{eq:dotprobdef_constraints}). Since $\bm{H}\succ\bm{0}$, every
principal submatrix $\bm{H}_{\alpha\alpha}$ is also positive
definite. Hence the objective (\ref{eq:dotprobdef}) is strictly
convex and coercive in $\dot{\bm{\beta}}_\alpha$, so it attains
a unique minimum over the nonempty closed feasible polyhedron.
\end{proof}

Well-posedness alone does not identify the minimiser with
the direction of the solution path. The following result
establishes this correspondence for the full directional
cone, allowing each boundary coordinate either to remain
zero or to move into its permitted sign. The restrictions
under which an implementation may hold a coordinate at
zero are stated subsequently.

\begin{proposition}[Local continuation]
\label{prop:local_continuation}
Let $\bm{H}\succ\bm{0}$, let $\lambda>0$, and let
$(\bm{\beta},\bm{\mu},\bm{g})$ be an exact KKT point of
(\ref{eq:CLASSOprobdef}) at $\lambda$. Let $\dot{\bm{\beta}}$ be the
unique minimiser of the directional problem
(\ref{eq:dotprobdef})--(\ref{eq:dotprobdef_constraints}) over the full
cone. Then there exists $\varepsilon>0$ such that for every
$0\le\Delta<\varepsilon$ the point
$\bm{\beta}-\Delta\,\dot{\bm{\beta}}$ is feasible and is the
\emph{unique global minimiser} of (\ref{eq:CLASSOprobdef}) at
$\lambda-\Delta$.
\end{proposition}

\begin{proof}
See Appendix~\ref{sec:branch_uniqueness}.
\end{proof}

The proposition applies to an exact certificate, including
cases with redundant rows, simultaneous events or non-unique
multipliers. Moreover, the continued point is the unique
minimiser among all feasible points, rather than merely
among affine continuations with a prescribed state.
Directional multipliers can be constructed by a finite
Farkas argument and need not be assumed separately.
If two certificates at the same primal point gave different
outgoing directions, local uniqueness would be violated;
hence the primal direction is independent of the certificate
selection. The certificate does, however, affect the step
allowed by the event scan.

One may now ask when a boundary coefficient can be held at
zero. Adding equations
$\dot{\bm{\beta}}_{\hat b}=\bm{0}$ for
$\hat b\subseteq b$ restricts the directional feasible set,
so the full-cone solution is retained precisely when it
satisfies these additional equations.

\begin{proposition}[Safe holds]
\label{prop:safe_hold}
Let $\dot{\bm{\beta}}$ be as above and let $\hat b\subseteq b$. The
restricted problem obtained by appending
$\dot{\bm{\beta}}_{\hat b}=\bm{0}$ has $\dot{\bm{\beta}}$ among its
minimisers if and only if $\dot\beta_j=0$ for every $j\in\hat b$, in
which case (since $\bm{H}\succ\bm{0}$) any exact solve of the
restricted problem returns $\dot{\bm{\beta}}$ itself. If in addition
$\beta_j=0$ for $j\in\hat b$, those coordinates remain zero
throughout the outgoing segment.
\end{proposition}

\begin{remark}

A feasible hold need not preserve the optimum. For example,
take $\bm{H}=\bm{I}_2$, $\bm{h}=(2,1)^{\rm t}$ and
$\beta_1\le3/2$. At $\lambda=1-\delta$, with
$0<\delta<1/2$, holding $\beta_2=0$ is feasible, but
increases the objective by $\delta^2/2$ relative to
$\beta_2=\delta$. Thus the velocity condition in
Proposition~\ref{prop:safe_hold} cannot be replaced by
primal feasibility of the held state.

\end{remark}

\begin{remark}

The release refinement in
Section~\ref{sec:release_refinement} uses the full-cone QP
to propose a held set, then tests the associated KKT signs
and changes that set when necessary. If the refinement
satisfies the full directional KKT conditions, the resulting
holds obey Proposition~\ref{prop:safe_hold}. Numerical
thresholds and fallback cases are described separately in
Section~\ref{sec:implementation}.

\end{remark}

For comparison with parametric-QP theory, the original
problem may be written in the equivalent form
\begin{equation}
\min_{\bm{\beta},\bm{t}}
\tfrac{1}{2}\bm{\beta}^{\rm t}\bm{H}\bm{\beta}
-\bm{h}^{\rm t}\bm{\beta}
+\lambda \bm{1}^{\rm t}\bm{t}
\label{eq:theory_split_qp_obj}
\end{equation}
subject to
\begin{equation}
-\bm{t}\le \bm{\beta}\le \bm{t},\qquad
\bm{A}_{e:}\bm{\beta}=\bm{b}_e,\qquad
\bm{A}_{i:}\bm{\beta}\le \bm{b}_i,
\label{eq:theory_split_qp_constraints}
\end{equation}
where $\bm{A}_{e:}$ and $\bm{A}_{i:}$ denote the equality and
inequality rows of $\bm{A}$ in (\ref{eq:CLASSOprobdef}). The two
problems have the same optimal $\bm{\beta}(\lambda)$ by taking
$t_i=|\beta_i|$ (necessarily so for $\lambda>0$), so any results established for
(\ref{eq:theory_split_qp_obj})--(\ref{eq:theory_split_qp_constraints})
apply directly to the native CLASSO path. In the split form the
problem is a parametric QP in the scalar parameter~$\lambda$, with
linear dependence on $\lambda$ only through the objective.

Let $\bm{x}=(\bm{\beta},\bm{t})\in\mathbb{R}^{2p}$, let
$\bm{E}\bm{x}=\bm{f}$ denote the equality constraints in
(\ref{eq:theory_split_qp_constraints}), and let
$\bm{C}\bm{x}\le\bm{d}$ denote the inequality constraints. For any
candidate active inequality set $J$, define the active
constraint matrix $\bm{C}_{J:}$ and let
$\bm{N}_{J}$ be any matrix whose columns form a basis for
the nullspace of \(
\left(\bm{E}^{\rm t}\,\,\, \bm{C}_{J:}^{\rm t}\right)^{\rm t}\).
The set $J$ subsumes both the binding-inequality index $z$
of Section~\ref{sec:new_approach} (via the $m_i$ original-problem
rows) and the native active/inactive decomposition
$(a,\bar{a},b)$ together with the sign pattern of $\bm{\beta}$
(via the $2p$ envelope constraints $\pm\beta_k\le t_k$). This
split-form bookkeeping is used only to state
\textbf{(B2)}--\textbf{(B4)}; the theorem and the bound below are
proved in the native coordinates.

\medskip

We distinguish feasibility from three commonly used
nondegeneracy conditions. Only \textbf{(B1)} is required
for the structural theorem; \textbf{(B2)}--\textbf{(B4)}
are included to describe the relation to conventional
active-set analyses and the numerical margins examined
in Section~\ref{sec:nondegeneracy}.

\begin{enumerate}
\item[\textbf{(B1)}] The feasible set of
(\ref{eq:theory_split_qp_obj})--(\ref{eq:theory_split_qp_constraints})
is nonempty for all $\lambda\in[0,\lambda_{\max}]$.
\item[\textbf{(B2)}] For every active set $J$ that occurs
along the solution path, the active constraints satisfy LICQ, i.e.\
the rows of  \(
\left(\bm{E}^{\rm t}\,\,\, \bm{C}_{J:}^{\rm t}\right)^{\rm t}\)
are linearly independent.
\item[\textbf{(B3)}] For every such $J$, the reduced Hessian
\(
\bm{N}_{J}^{\rm t}
\left(\begin{array}{cc}
\bm{H} & \bm{0}\\
\bm{0} & \bm{0}
\end{array}\right)
\bm{N}_{J}
\)
is positive definite.
\item[\textbf{(B4)}] Strict complementarity holds on every open path
segment.
\end{enumerate}

The split formulation should not include redundant bounds
$\bm{t}\ge\bm{0}$ when assessing \textbf{(B2)}. At a zero
coefficient, the two envelope rows already imply this
bound; adding it produces three active gradients
$(1,-1)$, $(-1,-1)$ and $(0,-1)$ in the
$(\beta_j,t_j)$ plane, which are dependent.
With the envelope rows alone, strict complementarity at a
zero coefficient corresponds to $|g_j|<1$.

For $\lambda>0$, optimality gives
$\bm{t}=|\bm{\beta}|$. At least one envelope row is
therefore active for each coordinate, so a null direction
with zero $\bm{\beta}$ component must also have zero
$\bm{t}$ component. Together with $\bm{H}\succ\bm{0}$,
this gives the reduced-Hessian condition \textbf{(B3)}
at the optimum. At $\lambda=0$, the objective is
independent of $\bm{t}$ and the split solution need not
be unique; the structural theorem concerns
$\bm{\beta}(\lambda)$, or equivalently the canonical
choice $\bm{t}=|\bm{\beta}|$.

The structural result itself is proved in the native coordinates,
where no constraint qualification is needed and the split form plays
no part. Call a \emph{pattern} an assignment giving each coordinate
$j\in\{1,\dots,p\}$ one of the three states
\[
\beta_j=0,\ |g_j|\le1;
\qquad
\beta_j\ge0,\ g_j=1;
\qquad
\beta_j\le0,\ g_j=-1,
\]
together with a subset of $\{1,\dots,m_i\}$ declared binding. There
are exactly $3^{p}2^{m_i}$ patterns, and every certified point
realises one. For a fixed pattern, write the stationarity condition
with the \emph{scaled} subgradient $\bm{z}=\lambda\bm{g}$ in place of
$\bm{g}$: the resulting system is linear in
$(\bm{\beta},\bm{\mu},\bm{z},\lambda)$ jointly, and remains meaningful at
$\lambda=0$, where converting back to $\bm{g}$ would divide by
zero.

\begin{theorem}[Piecewise-affine path with finitely many kinks]
\label{thm:piecewise_affine}
Let $\bm{H}\succ\bm{0}$ and let the feasible set be nonempty
(\textbf{(B1)}). Then the solution path
$\lambda\mapsto\bm{\beta}(\lambda)$ of (\ref{eq:CLASSOprobdef}) is
single-valued and continuous on $[0,\infty)$, the finitely many
pattern domains cover $[0,\infty)$, and $\bm{\beta}$ is affine on
each. In particular $\bm{\beta}$ has only finitely many kinks.
\end{theorem}

\begin{proof}
Since $\bm{H}\succ\bm{0}$ the objective is strictly convex and
coercive on the feasible set, so at each $\lambda\ge0$ a minimiser
exists and is unique; a KKT certificate exists at every such point,
by a Farkas argument requiring no constraint qualification. Each
certified point realises some pattern, so the pattern domains cover
$[0,\infty)$.

Fix a pattern and let $D$ be the set of $\lambda\ge0$ realising it.
The scaled system described above is linear in
$(\bm{\beta},\bm{\mu},\bm{z},\lambda)$, and the pattern's remaining conditions, namely sign conditions on
coordinates, nonnegativity of multipliers, and the inequalities
$|z_j|\le\lambda$, are linear inequalities in the same variables. Hence the set of admissible
certificates is convex, and so is its projection $D$: a subinterval of
$[0,\infty)$. If $\lambda_1,\lambda_2\in D$ with certified
minimisers $\bm{\beta}_1,\bm{\beta}_2$, then any convex combination
of the two certificates certifies the corresponding convex
combination of the minimisers for the same pattern; by uniqueness of
the minimiser, $\bm{\beta}$ agrees with that combination, so
$\bm{\beta}$ is affine on $D$. There are $3^{p}2^{m_i}$ patterns, so
$[0,\infty)$ is covered by finitely many intervals on each of which
$\bm{\beta}$ is affine, and the kinks are confined to their finitely
many endpoints. Continuity, including at $\lambda=0$, follows from
uniqueness of the minimiser together with the continuity of the
solution map of a strictly convex parametric quadratic programme
\citep[see, e.g.,][]{Best1996}.

The argument uses no constraint qualification, no strict
complementarity, and no uniqueness of the multipliers; redundant and
duplicated constraint rows are permitted, as are simultaneous
boundaries. It is a statement about $\bm{\beta}(\lambda)$, which is
what the native problem requires; no uniqueness is claimed for the
auxiliary split variable $\bm{t}$ at $\lambda=0$.
\end{proof}

The count of patterns bounds the number of segments directly.

\begin{corollary}[Finite combinatorial bound]
\label{cor:complexity}
Let $\bm{H}\succ\bm{0}$ and let the feasible set be nonempty. The
number of affine path segments satisfies
\begin{equation}
N_{\rm seg}\le 3^{p}\,2^{m_i},
\label{eq:complexity_bound}
\end{equation}
and hence the number of kinks satisfies
\begin{equation}
N_{\rm kink}\le 3^{p}\,2^{m_i}-1.
\label{eq:kink_bound}
\end{equation}
\end{corollary}

\begin{proof}
Immediate from Theorem~\ref{thm:piecewise_affine}: there are
$3^{p}2^{m_i}$ patterns and each contributes at most one interval of
the cover, on which the path is affine.
\end{proof}

For comparison, counting active sets in the split formulation
under LICQ gives
$\sum_{r\le2p-m_e}\binom{m_i+2p}{r}$.
The native count avoids this qualification and expresses
the bound directly in terms of coefficient signs and
binding inequalities. Neither count is intended as a
practical estimate. For example, with $p=10$ and $m_i=10$,
(\ref{eq:complexity_bound}) is approximately
$6\times10^7$, whereas the corresponding diabetes paths
have only a few kinks. The observed counts in
Section~\ref{sec:scaling} grow approximately linearly
with $p$ for the tested configurations.

An exponential worst-case dependence cannot be excluded
for general exact path-following: the unconstrained LASSO
is a special case, and \citet{Mairal2012} constructs
examples with exponentially many segments. Suppressing
events \textbf{E1} and \textbf{E4}, so that coefficients
and binding rows can enter but cannot leave, instead
gives a monotone event count of order $p+m_i$.
Such a restriction generally changes the curve being
followed and therefore does not provide a polynomial
bound for the constrained-LASSO path.

\subsection{Termination of the traversal}\label{sec:termination}

A finite number of primal path segments does not imply
termination of every procedure that propagates valid
certificates along them. The procedure may record events
at which the certificate changes while the primal path
continues in the same direction. Thus a termination
argument must control certificate selection and progress,
in addition to establishing a correct local direction.

For Algorithm~\ref{alg:classo} that control is supplied by
following the canonical certificate. The step it takes is
characterised as follows.

\begin{proposition}[The certificate-following step]
\label{prop:efficient_step}
Let $\bm{H}\succ\bm{0}$, let $\lambda^{(k)}>0$, let $\bm\beta^{(k)}$ be
the solution at $\lambda^{(k)}$, let $\dot{\bm\beta}$ be the minimiser
of the directional QP (\ref{eq:dotprobdef})--(\ref{eq:dotprobdef_constraints}),
and let $\bm\mu^\star$ and $\bm w^\star$ be given by
(\ref{eq:canonical_certificate}) and (\ref{eq:lex_rate}). In exact
arithmetic:
\begin{enumerate}
\item[(i)] the set of $\Delta\in[0,\lambda^{(k)}]$ for which the
outgoing pattern holds at $\bm\beta^{(k)}-\Delta\dot{\bm\beta}$,
$\bm\mu^\star+\Delta\bm w^\star\in
\mathcal{M}(\bm\beta^{(k)}-\Delta\dot{\bm\beta},\lambda^{(k)}-\Delta)$
and (\ref{eq:cert_horizon}) holds is a closed interval
$[0,\bar\Delta]$ with $\bar\Delta>0$;
\item[(ii)] $\bar\Delta$ is the smallest positive root of
\textbf{E1}--\textbf{E6}, or $\lambda^{(k)}$ if that is smaller,
whichever and however many of the conditions attain it;
\item[(iii)] on $[0,\bar\Delta]$, $\bm\mu^\star+\Delta\bm w^\star$ is the
canonical certificate at
$(\bm\beta^{(k)}-\Delta\dot{\bm\beta},\lambda^{(k)}-\Delta)$, with
derivative $\bm w^\star$ on $(0,\bar\Delta)$ and one-sided derivatives
$\bm w^\star$ at its ends, and $\lambda^{(k)}-\bar\Delta$ is the end of
a cell of the canonical certificate: a set of $\lambda$, fixed
independently of $\lambda^{(k)}$, on which the outgoing pattern holds
and the canonical certificate lies on a given face of the certificate
polyhedron.
\end{enumerate}
\end{proposition}

\begin{proof}[Sketch]
(i) The pattern and feasibility conditions are affine inequalities in
$\Delta$, and (\ref{eq:cert_horizon}) asks an affine path to lie in a
closed convex cone, so each defines a closed interval containing $0$;
the intersection has positive length because an outgoing step always
exists. (ii) \textbf{E1} and \textbf{E3} are the roots of the pattern
inequalities, \textbf{E2}, \textbf{E4} and \textbf{E5} those of
feasibility, and \textbf{E6} that of (\ref{eq:cert_horizon}).
(iii) Along any cell starting at $\bm\mu^\star$, the derivative of the
canonical certificate is a lexicographic rate, and that rate is
unique, so it is $\bm w^\star$. The cell of the rows
$\mathcal{T}^\star$ contains no canonical certificate off the line
$\bm\mu^\star+\Delta\bm w^\star$, so the line leaves that cell only
where the canonical certificate leaves it. The admissible interval
therefore ends at the first $\Delta$ beyond which the pattern,
feasibility or optimality condition prevents continuation; each set
is closed, so the conditions still hold at $\bar\Delta$ itself. The
statements are machine-checked
(Section~\ref{sec:formal}).
\end{proof}

Because $\dot{\bm\beta}$, $\bm\mu^\star$, $\bm w^\star$ and $\bar\Delta$
are each unique, the step is determined: simultaneous events need no
choice among them, and the next step recomputes all four at the new
point.

\begin{theorem}[Termination in exact arithmetic]
\label{thm:termination}
Let $\bm{H}\succ\bm{0}$ and let the feasible set be nonempty. Consider
Algorithm~\ref{alg:classo} in exact arithmetic with $\lambda_{\min}=0$,
no iteration limit and no minimum step, each step being computed
exactly as in Proposition~\ref{prop:efficient_step}. From any
certified point, the run reaches $\lambda=0$ within
\begin{equation}
N_{\rm step}\le\left(3^{p}\,2^{m_i}\right)^2\,2^{2p+2m_i}
\label{eq:termination_bound}
\end{equation}
steps. Started from the lexicographic initial point of
Section~\ref{sec:init}, the run traverses the whole path: every
recorded point is the solution at its value of $\lambda$, and on
each step the affine interpolation between recorded points is the
solution path.
\end{theorem}

\begin{proof}[Sketch]
By Proposition~\ref{prop:efficient_step}(iii), every step ends at the
end of a cell, and its endpoint certificate is canonical, so the next
step starts from the canonical certificate again. A cell is indexed by
the outgoing pattern, of which Corollary~\ref{cor:complexity} gives at
most $(3^p2^{m_i})^2$ counting the primal formula and the pattern
separately, and a face of the certificate polyhedron, of which there
are at most $2^{2p+2m_i}$. Cells are fixed sets of $\lambda$, so the
number of cell ends below the current $\lambda$ is a nonnegative
integer, bounded by (\ref{eq:termination_bound}), that every step
strictly reduces. No minimum step length, general position or record
of visited cells is assumed. The full argument is machine-checked
(Section~\ref{sec:formal}).
\end{proof}

The stopping controls of Algorithm~\ref{alg:classo} lie outside the
theorem by design: a positive $\lambda_{\min}$ ends the run earlier,
while the iteration limit, the minimum step and the substitutions of
Section~\ref{sec:follow} exist for floating-point arithmetic, and the
implementation reports whenever any of them acts. Removing points
that are not kinks changes which points are recorded, not the steps
taken.

The bound is exponential, as any worst-case bound must be
(Section~\ref{sec:theory}); its role is to establish termination,
not to estimate cost. It is the certificate rule that matters, not
only the primal direction. A natural alternative propagates a
minimum-norm rate between kinks and recovers a new certificate at
each one. It traces the same primal path, but it need not
terminate, as three examples show.

A simple example has $p=1$, $\bm{H}=\bm{1}$,
$\bm{h}=\bm{0}$ and the equality $\beta=0$.
The primal path is constant. At $|g|=1$, however, a
directional certificate may choose $\dot\beta=0$,
$q=2g$ and $\dot\nu=-2g$. The subgradient then reaches
the opposite boundary after a step $2\lambda/3$.
Repeating the construction gives the infinite sequence
$\lambda/3^k$, even though every step is positive and
every propagated certificate is valid.

A second, machine-checked counterexample satisfies the
initialisation and minimum-norm conventions in full. It
starts at the lexicographic minimum with the smallest
certifying $\lambda$, chooses minimum-norm multipliers and
minimum-norm directional rates, and takes the largest step
allowed by the subgradient bound. Re-certification at the
endpoint reverses the preceding subgradient progress,
producing $\lambda_k=7^{-k}$. A third, also machine-checked,
isolates the minimum-norm solution of the fixed-state system
(\ref{eq:segment_kkt}): on the same problem it assigns zero rate
to a held multiplier, while every valid certificate at the point
has a nonzero multiplier on that row, so re-certification
restores a positive one at each step and the parameters contract
as $\lambda_k=3^{-k}$. In both, the step ends at a point that is
not a kink, and recovering a fresh certificate there undoes the
progress the previous step made. Following the canonical
certificate removes that mechanism: on the $7^{-k}$ problem
Algorithm~\ref{alg:classo} traverses the single segment in one
step.

The software retains the alternative policy as an option, under
the name \emph{recertify}, for comparison. On the $7^{-k}$
problem it records $17$ events at the default $\lambda_{\min}$,
and $22$ with the floor reduced to $10^{-30}$, all subdividing a
single primal segment, whereas Algorithm~\ref{alg:classo} records
the two endpoints and reaches $\lambda=0$ in both cases.

This distinction also limits what a tie-breaking rule
can establish. \citet{Schmelzer2026} uses a lowest-index
(Bland) rule in the ratio test, while also discussing the
remaining termination question. The counterexamples
above concern certificate-rate selection and can have
only one available event at a step, so resolving ties
between event indices does not address their mechanism;
following the canonical certificate does.

For a run with $N_{\rm step}$ iterations, the local
solve count is proportional to $N_{\rm step}$ when the
number of refinement and recovery solves per iteration
is bounded. Theorem~\ref{thm:termination} bounds
$N_{\rm step}$, while Corollary~\ref{cor:complexity}
bounds the affine pieces of the primal path. These counts
should therefore be distinguished both in complexity
claims and when reporting numerical results.
The implementation records a traversal status and a
certificate-failure count separately, as described in
Appendix~\ref{sec:safeguards}.

The numerical refinements below do not change the
mathematical path. Their role
is to compute its local segments reliably in finite
precision. Sections~\ref{sec:endpoint}
and~\ref{sec:segment_verification} assess that accuracy,
while Section~\ref{sec:nondegeneracy} reports margins
associated with the active-set classifications.

\subsection{Formal verification}\label{sec:formal}

The theoretical results have been formalised in Lean~4
against Mathlib. The sources and pinned dependency revision
are distributed with \REPOWHERE.\REPOFOOTNOTE{}
The development reports no admitted proofs and uses no
axioms beyond Lean's standard three.

The checked results include existence of a minimiser and
KKT certificate for every $\lambda\ge0$; equivalence
of the reduced directional QP and the full directional
cone; continuity and global Lipschitz continuity of the
primal path; its piecewise-affine structure and finite
cover; local continuation and safe holds; and the
lexicographic initialisation, including attainment of the
smallest certifying $\lambda$ and uniqueness of the
minimum-norm multiplier selection. The global results
assume only $\bm{H}\succ\bm{0}$ and a nonempty feasible
set, permitting dependent rows and simultaneous boundaries.

Proposition~\ref{prop:efficient_step} and
Theorem~\ref{thm:termination} are checked, as are the
counterexamples for unsafe holds and the two nonterminating
witnesses for the re-certifying policy in
Section~\ref{sec:termination}. In the formal statements a step
is defined exactly as in Proposition~\ref{prop:efficient_step}:
the minimum-norm multiplier vector over all rows, the directional
minimiser, the lexicographic rate and the greatest admissible
$\Delta$. The file \texttt{formal/README.md} of the release maps
each result of this section to its Lean declaration.

The formalisation concerns exact arithmetic. It makes no claim
about floating-point arithmetic, the Python implementation or the
external QP solver. That the implementation follows the exact rule
is assessed numerically instead, by a step-end check of the
propagated certificate and by comparison with an exact rational
implementation (Sections~\ref{sec:segment_verification}
and~\ref{sec:prototype}).

Related Lean~4 developments include the smooth KKT
formalisation of \citet{Li2025}, the Farkas and
linear-programming duality results of
\citet{Dvorak2026}, and the convex-analysis and
first-order-method infrastructure in the
\texttt{optsuite} project's \texttt{optlib}.
At the revisions examined, the required finite-cone
multiplier statement was not available in a directly
usable form, so it is proved locally behind a small
interface. The directional reformulation, continuation,
safe-hold, finite-cover and initialisation results
are specific to the present development.

\section{Numerical implementation}\label{sec:implementation}

The numerical implementation maintains the active,
boundary and binding sets, the strong/weak partition
of binding inequalities, and the signs of the active
coefficients. Once the outgoing state has been
determined, its direction follows from a linear KKT
system, the certificate's rate and horizon from two small
linear programmes, the candidate event times are roots of
affine functions, and the certificate at the next point is
the canonical one, computed as a small least-distance
problem.
This structure allows the accuracy of propagation to
be improved without requiring every QP solve to the
same precision.

We first use the directional QP to propose which
boundary coordinates move and which weakly active
rows remain tight. This decision is refined through
KKT sign tests, after which the continuous quantities
are recomputed for the selected state. Throughout this
section, a linear solve is understood to be performed
in floating point: its accuracy depends on conditioning
and rounding, and the acceptance tests use finite
tolerances. Appendix~\ref{sec:safeguards} records
the failure mechanisms that motivated these choices.

\subsection{The release decision and its refinement}
\label{sec:release_refinement}

The directional QP is solved numerically with a nominal
tolerance of $10^{-5}$. Its output is used to identify
released coordinates and tight weakly active rows,
rather than propagated directly. In our tests,
propagation of the unrefined direction displaced event
times by relative amounts of order $10^{-4}$, which
was sufficient to make the state at the next point
inconsistent.

Let $f$ contain the active coefficients and the
members of $b$ proposed for release, and let
$\bm{R}$ contain the equality rows, strongly active
inequalities and weakly active rows proposed to remain
tight, restricted to $f$. The corresponding direction
and rates satisfy
\begin{equation}
\left(\begin{array}{cc}\bm{H}_{ff} & \bm{R}^{\rm t}\\ \bm{R} & \bm{0}\end{array}\right)
\left(\begin{array}{c}\dot{\bm{\beta}}_f\\ \dot{\bm{\mu}}\end{array}\right)
=
\left(\begin{array}{c}-\bm{g}_f\\ \bm{0}\end{array}\right),
\label{eq:segment_kkt}
\end{equation}
with $\dot\beta_j=0$ on the held members of $b$. When the tight rows
are dependent (duplicated or redundant constraints) the dual block is
rank-deficient; $\dot{\bm{\beta}}_f$ is nevertheless unique because
$\bm{H}_{ff}\succ\bm{0}$. The minimum-norm multipliers of this system
serve the sign tests below; the multiplier rates propagated along
the segment are those of Section~\ref{sec:follow}.

The proposed state is checked using the directional KKT
conditions. For a held coordinate,
$\rho_j=-g_j(\bm{H}\dot{\bm{\beta}}+\bm{g}
+\bm{R}^{\rm t}\dot{\bm{\mu}})_j$ must be non-negative.
A negative value therefore requires release; conversely,
a released coordinate whose direction violates its sign
cone must be held. The analogous tests require a tight
weakly active row to have the correct multiplier sign,
and a released row to remain directionally feasible.
The state is updated when a test fails and
(\ref{eq:segment_kkt}) is solved again.

When this iteration satisfies all directional KKT
conditions, strict convexity identifies its direction
with the full-cone minimiser. In the reported experiments
the refinement generally converges in a few steps.
A cycle guard instead retains the original QP solution
if refinement fails; this fallback does not acquire the
accuracy of a successful refinement merely by being
returned. The sign tolerances are given in
Section~\ref{sec:residual_tolerances}.

The refinement can correct a mistaken release decision
as well as improve the continuous solution. For example,
on one suite case a direction component of order
$10^{-4}$, within a vector of norm $10^2$, was initially
classified as a hold. Without the sign refinement the
corresponding coordinate did not enter. Recomputing
the direction for a verified state also reduces errors
in subsequent event times.

\subsection{The certificate at a kink}\label{sec:exact_certificate}

At a kink the state $(a,z)$ and the signs $\bm{s}_a$ are declared, and
the primal--dual pair is the solution of the fixed-state KKT system
\begin{equation}
\left(\begin{array}{cc}\bm{H}_{aa} & \bm{A}_{za}^{\rm t}\\ \bm{A}_{za} & \bm{0}\end{array}\right)
\left(\begin{array}{c}\bm{\beta}_a\\ \bm{\mu}_z\end{array}\right)
=
\left(\begin{array}{c}\bm{h}_a-\lambda\bm{s}_a\\ \bm{b}_z\end{array}\right).
\label{eq:kink_kkt}
\end{equation}

Its primal component, checked against the declared coefficient
signs and constraint feasibility, replaces the propagated
$\bm{\beta}$: each accepted coefficient vector is recomputed at its
current $\lambda$, so propagation error does not accumulate.

The certificate recorded at the point is the canonical one, the
minimum-norm point of the polyhedron of valid multipliers at
$(\bm\beta,\lambda)$. Its equalities, stationarity on the active
coordinates, are solved by least squares; if that point satisfies
the inequalities (the inactive subgradient bounds and the signs of
the inequality multipliers), it is the answer. Otherwise the
problem is posed on the null space of the equalities as a
least-distance problem and solved by Lawson and Hanson's reduction
to non-negative least squares. That solver's answer is checked
against its own optimality conditions, since a library
implementation was observed to return a non-optimal point without
signalling an error, and the same problem is then re-solved by
bounded-variable least squares. A QP solve to $10^{-12}$ is the last
resort. The rows that the solution makes tight are finally re-solved
exactly as linear equations. The bounds are relaxed by the
acceptance tolerance, because at a kink the polyhedron has empty
interior and sits within rounding of infeasibility. If no canonical
certificate can be computed, the fixed-state certificate from
(\ref{eq:kink_kkt}) is used instead and the substitution is
reported with the output. Either way, the certificate must pass the
same residual-based acceptance criteria.

An entry event also places its coordinate in $b$: with dependent
bound rows, different certificates can assign different
subgradients at the same point. The resulting state and certificate
must nevertheless pass the acceptance checks; the event history
alone is not a certificate.

\subsection{Coordinates held by bounds, and duplicated rows}\label{sec:one_sided}

For a coefficient with a single one-sided sign bound in row $i$,
$|\beta_j|=s_j\beta_j$ over the feasible set, so stationarity fixes
only the combination $\lambda g_j+a_{ij}\mu_i$. The subgradient and
the bound multiplier are not separately determined: different
splits are different valid certificates. Any rule that does not fix
the split lets the certificate redistribute between them at every
kink, recording events at which the primal path does not bend. The
canonical certificate fixes it, as the split of least norm, and so
needs no further rule. Exactly duplicated constraint rows are
treated in the same way: their common multiplier is shared equally,
which is the minimum-norm division.

On the stress suite of Section~\ref{sec:stress}, every interior
recorded point is a change of primal direction ($502$ of $502$).
Under the re-certifying policy of Section~\ref{sec:termination}
with the split left free, $151$ of $655$ interior points are not,
in $11$ sign-constrained cases; that policy needs an explicit rule
assigning the whole combination to the bound multiplier to suppress
them. For unrestricted coefficients, and for boxes that permit both
signs, the subgradient alone holds the coordinate and \textbf{E5}
remains necessary.

\subsection{Following the certificate}\label{sec:follow}

The lexicographic rate is computed in two stages. A linear programme
gives the fastest rate of decrease of the certificate's norm over
the rates that keep it feasible to first order; the rate of least
norm among those that achieve it is then found as a least-distance
problem, by the same method as the certificate itself. A second
linear programme gives the certificate horizon, the largest step
over which this rate keeps the certificate feasible and minimum-norm
on the rows it keeps tight; a horizon shorter than the step reached
by the primal events is \textbf{E6}. A horizon below the minimum
step is ignored. Both programmes are small: their variables are the
multipliers of the binding rows.

Each of these solves can fail. If the second stage of the rate
fails, the optimal vertex of the first is used; if the first fails,
the minimum-norm rate is used; if the horizon programme fails, no
\textbf{E6} is scheduled on that segment. The implementation reports
every such substitution with the output, together with any
substitution of the fixed-state certificate in
Section~\ref{sec:exact_certificate}, and warns when one occurs: the
affected points are still certified, but they were not computed by
the rule of Theorem~\ref{thm:termination}. On the benchmarks
reported below, none occurs.

A point at which neither \textbf{E1} nor \textbf{E3} occurred and
across which the primal direction is unchanged, to $10^{-8}$
relative, is removed from the output when the next direction is
known, as in Algorithm~\ref{alg:classo}. The output is unchanged as a
piecewise-affine function; the number of such steps is kept among
the run's statistics, so that work is not mistaken for recorded
points.

\subsection{What remains tolerance-dependent}
\label{sec:residual_tolerances}

The release-refinement sign tests use a tolerance of
$10^{-10}$ relative to the direction norm. Boundary
membership is identified using
$|g_j|>1-\epsilon_{\rm b}$, with
$\epsilon_{\rm b}=10^{-8}$, supplemented by the
entry-event information of
Section~\ref{sec:exact_certificate}. These
thresholds must be consistent with the accuracy of
the quantities being classified. In the examples of
Appendix~\ref{sec:safeguards}, reducing
$\epsilon_{\rm b}$ to $10^{-10}$ excluded a
coordinate whose boundary status could not be resolved
at that precision and changed the computed path.

Admitting a near-boundary coordinate allows the
directional QP and subsequent checks to assess its
behaviour; excluding it may delay its release in
the implemented event logic. This asymmetry motivates
a boundary tolerance above the estimated certificate
error, but does not eliminate the need to verify the
resulting segments.

For $n<p$, all stages use the same ridge-augmented
Hessian, so the resulting path solves the modified
problem of Section~\ref{sec:problem}. Comparisons with matched and unregularised numerical
reference solutions are reported separately below.
Certificate checks use the unclipped inactive
subgradient, expressed as $\lambda\bm{g}$ and
normalised by $\lVert\bm{h}\rVert_\infty$, together
with the active stationarity residual. Clipping the
subgradient before testing its bound would conceal
the violations that the check is intended to detect.

\section{Synthetic experiments}\label{sec:synthetic}

We first assess numerical accuracy on a synthetic suite
containing several constraint geometries and conditioning
regimes, then examine the computational cost as $p$ increases.
The diabetes examples in Section~\ref{sec:diabetes_app}
provide comparisons with a discrete $\lambda$-sweep and
the local linear-system construction of \citet{Gaines2018}.

\subsection{Synthetic stress suite}\label{sec:stress}

The stress suite contains $39$ cases in six families,
chosen to test equality and inequality constraints,
redundant rows, near-degeneracy and ridge augmentation:

\begin{enumerate}
\item[(i)] \textit{Equality only} ($6$ cases): $m_e$ random equality
constraints with $p\in\{5,8,10\}$, $m_e\in\{1,2,3,4\}$, varying
condition number.
\item[(ii)] \textit{Inequality only} ($8$ cases): sign, box, and
random inequality constraints with $p\in\{5,8,10\}$.
\item[(iii)] \textit{Mixed/redundant} ($7$ cases): combinations of
equality and inequality constraints, some with deliberately
duplicated inequality rows.
\item[(iv)] \textit{Nearly degenerate} ($6$ cases): designs with
clustered singular values in $\bm{X}$, producing near-tied kink
events.
\item[(v)] \textit{Wide} ($6$ cases, $n<p$): problems in which the
Hessian is singular before ridge augmentation, with $p\in\{10,15,20\}$
and $n\in\{6,8,10\}$.
\item[(vi)] \textit{Larger} ($6$ cases): moderate-scale problems with
$p\in\{20,30,50\}$.
\end{enumerate}

Each problem is generated from a specified random seed.
An SVD of a Gaussian matrix is used to construct
$\bm{X}$ with the required condition number, and a sparse
coefficient vector with approximately $p/3$ non-zero
elements of unit scale generates
$\bm{y}=\bm{X}\bm{\beta}_{\rm true}+\epsilon$,
where $\epsilon\sim\mathcal{N}(0,\sigma^2\bm{I})$
and $\sigma$ is small. The modest dimensions permit
individual discrepancies to be investigated in detail.
Seeds $60$ and $73$ have $\lambda_{\max}=0$ and hence
trivial paths; the remaining $37$ require traversal.

\subsection{Correctness along the path}\label{sec:endpoint}

Agreement at $\lambda=0$ alone is insufficient to assess
a path-following method, since an incorrect intermediate
path can still reach the constrained least-squares
solution. We therefore solve the constrained LASSO
directly in split-variable form, to a nominal QP
tolerance of $10^{-10}$, at $25$ logarithmically
spaced values between $0.95\lambda_{\max}$ and
$10^{-4}\lambda_{\max}$. At each value we compare
the reference solution with the affine interpolation
of the CLASSO output, using the signed relative objective
difference
$(f(\bm{\beta}_{\rm CLASSO})-f(\bm{\beta}_{\rm ref}))/
|f(\bm{\beta}_{\rm ref})|$, where $\bm{\beta}_{\rm ref}$ denotes the numerical reference solution.
The endpoint difference is evaluated separately at
$\lambda=0$, as the signed unnormalised objective difference
$f(\bm{\beta}_{\rm CLASSO})-f(\bm{\beta}_{\rm ref})$. These are signed comparisons between numerical solutions,
rather than certified errors relative to the exact optimum.
A negative difference means that CLASSO returns a lower
computed objective than the reference solve; this may reflect
reference-solve error, feasibility residuals or rounding, and
does not establish an improvement over the mathematical
optimum. The requested QP tolerance alone is not an error
bound on the optimum. We therefore interpret these differences
together with the feasibility and KKT residuals reported below.
In each matched ridge-augmented comparison the mathematical
primal minimiser is unique.

Table~\ref{tab:stress_path} reports the sampled and
endpoint differences, recorded point counts, and occurrences
of \textbf{E1} and \textbf{E5}. Across the $31$ nontrivial cases with
$n\ge p$, the signed relative objective differences at the sampled
$\lambda$ range from $-9\times10^{-15}$ to $3.9\times10^{-10}$;
none exceeds $10^{-9}$ in absolute value. At $\lambda=0$ the
signed unnormalised differences over the $33$ such cases range from
$-1.5\times10^{-7}$ to $2.0\times10^{-10}$, and two exceed
$10^{-9}$ in magnitude, both with CLASSO returning the lower
computed objective; relative to the reference objective, they
range from $-1.1\times10^{-6}$ to $5.4\times10^{-10}$. These comparisons
are interpreted alongside the independently checked feasibility
and KKT residuals; they do not assign the discrepancy to either
solver by themselves. The point count includes the endpoints.
The distinction from interior changes of direction
is relevant when comparing these counts with the
theoretical bound.

The table contains $578$ recorded points, of which $502$ are
interior to nontrivial paths, and each of these is a change of
primal direction at the numerical resolution tested
(Section~\ref{sec:one_sided}). A $300$-point grid was also used
for the cases with the largest counts.

For the six wide problems, CLASSO solves the
ridge-augmented objective. Table~\ref{tab:ridge} gives comparisons at five prescribed
ridge strengths. Each cell is a maximum over the sampled
$\lambda$, and the largest such matched difference over the
$30$ runs is $2\times10^{-9}$. A per-run maximum does not display
the opposite direction: over all sampled points of the same $30$
runs the differences reach $-3.6\times10^{-7}$.

The wide-case entries in
Table~\ref{tab:stress_path}, by contrast, compare
against the unregularised problem and measure
the effect of adding the ridge, subject also to numerical discrepancies in both solutions. This distinction
is necessary when interpreting accuracy for $n<p$.
On the machine and software stack specified in
Section~\ref{sec:benchmarking}, the six wide-case
outputs were bit-identical across five thread
counts and three separate processes.

\begin{table*}
\centering
\caption{Correctness along the path on the $39$-case stress suite.
Family as in Section~\ref{sec:stress}; `Points' is the number of
recorded points on the path, its two endpoints included; E1 and E5
are the numbers of active-leaving and held-coordinate sign-crossing
events on the true path; `path' is the
maximum over $25$ log-spaced $\lambda$ of the signed relative objective difference
between the CLASSO path and a direct constrained-LASSO reference solve
with requested tolerance $10^{-10}$; `endpoint' is the signed objective difference
$f(\bm{\beta}_{\rm CLASSO})-f(\bm{\beta}_{\rm ref})$ at $\lambda=0$, which is
unnormalised rather than relative. Differences are
rounded to one significant figure, so a cell may differ from the value
quoted in the text; entries
$<10^{-13}$ are reported at that threshold and are not certified error bounds. Seeds $60$ and $73$ have
$\lambda_{\max}=0$ and are trivial. Family (v) is $n<p$, and carries
the ridge of Section~\ref{sec:problem}; the differences shown for it are
against the \emph{unregularised} problem and therefore include the effect of
that ridge as well as numerical discrepancies. Matched comparisons are reported in
Table~\ref{tab:ridge}.\label{tab:stress_path}}
\smallskip
{\small
\CILFITWIDTH{%
\begin{tabular}{llrrrll}
\hline
Family & Case & Points & E1 & E5 & Path diff. & Endpoint diff.\\
\hline
(i) & $p$=5 $n$=20 $m_e$=1 $\kappa$=10 seed 42 & 5 & 0 & 0 & $<10^{-13}$ & $2\times10^{-13}$ \\
(i) & $p$=5 $n$=20 $m_e$=2 $\kappa$=10 seed 43 & 6 & 1 & 1 & $<10^{-13}$ & $1\times10^{-12}$ \\
(i) & $p$=8 $n$=30 $m_e$=1 $\kappa$=10 seed 44 & 8 & 0 & 0 & $<10^{-13}$ & $-1\times10^{-12}$ \\
(i) & $p$=8 $n$=30 $m_e$=3 $\kappa$=100 seed 45 & 8 & 1 & 1 & $<10^{-13}$ & $<10^{-13}$ \\
(i) & $p$=10 $n$=40 $m_e$=1 $\kappa$=10 seed 46 & 10 & 0 & 0 & $<10^{-13}$ & $<10^{-13}$ \\
(i) & $p$=10 $n$=40 $m_e$=4 $\kappa$=1000 seed 47 & 21 & 7 & 5 & $3\times10^{-13}$ & $<10^{-13}$ \\
(ii) & $p$=5 $n$=20 nonneg $\kappa$=10 seed 50 & 3 & 0 & 0 & $2\times10^{-11}$ & $7\times10^{-12}$ \\
(ii) & $p$=5 $n$=20 nonpos $\kappa$=10 seed 51 & 3 & 0 & 0 & $1\times10^{-10}$ & $1\times10^{-11}$ \\
(ii) & $p$=5 $n$=20 box $\kappa$=10 seed 52 & 8 & 0 & 0 & $1\times10^{-12}$ & $-2\times10^{-7}$ \\
(ii) & $p$=8 $n$=30 nonneg $\kappa$=100 seed 53 & 4 & 0 & 0 & $4\times10^{-11}$ & $2\times10^{-13}$ \\
(ii) & $p$=8 $n$=30 random $\kappa$=10 seed 54 & 11 & 1 & 0 & $<10^{-13}$ & $<10^{-13}$ \\
(ii) & $p$=10 $n$=40 nonneg $\kappa$=10 seed 55 & 5 & 0 & 0 & $4\times10^{-11}$ & $6\times10^{-13}$ \\
(ii) & $p$=10 $n$=40 box $\kappa$=1000 seed 56 & 28 & 4 & 0 & $7\times10^{-13}$ & $<10^{-13}$ \\
(ii) & $p$=10 $n$=40 random $\kappa$=100 seed 57 & 22 & 5 & 0 & $3\times10^{-13}$ & $<10^{-13}$ \\
(iii) & $p$=5 $n$=20 $m_e$=1 nonneg dup=0 $\kappa$=10 seed 60 & 1 & 0 & 0 & --- & $2\times10^{-11}$ \\
(iii) & $p$=8 $n$=30 $m_e$=1 nonneg dup=0 $\kappa$=10 seed 61 & 2 & 0 & 0 & $3\times10^{-11}$ & $-2\times10^{-13}$ \\
(iii) & $p$=8 $n$=30 $m_e$=2 box dup=0 $\kappa$=100 seed 62 & 18 & 5 & 1 & $<10^{-13}$ & $-5\times10^{-13}$ \\
(iii) & $p$=10 $n$=40 $m_e$=1 nonneg dup=3 $\kappa$=10 seed 63 & 7 & 0 & 0 & $5\times10^{-11}$ & $8\times10^{-12}$ \\
(iii) & $p$=10 $n$=40 $m_e$=2 random dup=0 $\kappa$=100 seed 64 & 20 & 5 & 0 & $<10^{-13}$ & $<10^{-13}$ \\
(iii) & $p$=10 $n$=40 $m_e$=1 nonneg dup=5 $\kappa$=10 seed 65 & 7 & 0 & 0 & $4\times10^{-11}$ & $-2\times10^{-9}$ \\
(iii) & $p$=10 $n$=40 $m_e$=2 box dup=2 $\kappa$=1000 seed 66 & 21 & 4 & 2 & $7\times10^{-13}$ & $<10^{-13}$ \\
(iv) & $p$=5 $n$=20 none/none seed 70 & 6 & 0 & 0 & $<10^{-13}$ & $<10^{-13}$ \\
(iv) & $p$=5 $n$=20 nonneg/none seed 71 & 4 & 0 & 0 & $8\times10^{-11}$ & $1\times10^{-11}$ \\
(iv) & $p$=8 $n$=30 none/sumzero seed 72 & 8 & 0 & 0 & $<10^{-13}$ & $-3\times10^{-13}$ \\
(iv) & $p$=8 $n$=30 nonneg/sumzero seed 73 & 1 & 0 & 0 & --- & $9\times10^{-11}$ \\
(iv) & $p$=10 $n$=40 nonneg/none seed 74 & 7 & 0 & 0 & $4\times10^{-10}$ & $5\times10^{-11}$ \\
(iv) & $p$=10 $n$=40 box/none seed 75 & 13 & 0 & 0 & $1\times10^{-11}$ & $6\times10^{-11}$ \\
(v) & $p$=10 $n$=6 none/none $\kappa$=10 seed 80 & 23 & 6 & 0 & $3\times10^{-2}$ & $4\times10^{-8}$ \\
(v) & $p$=10 $n$=6 nonneg/none $\kappa$=10 seed 81 & 10 & 2 & 0 & $1\times10^{-4}$ & $1\times10^{-6}$ \\
(v) & $p$=15 $n$=8 none/sumzero $\kappa$=10 seed 82 & 23 & 4 & 0 & $4\times10^{-2}$ & $5\times10^{-7}$ \\
(v) & $p$=15 $n$=8 nonneg/none $\kappa$=100 seed 83 & 12 & 2 & 0 & $2\times10^{-3}$ & $2\times10^{-4}$ \\
(v) & $p$=20 $n$=10 nonneg/none $\kappa$=10 seed 84 & 6 & 0 & 0 & $6\times10^{-6}$ & $4\times10^{-7}$ \\
(v) & $p$=20 $n$=10 box/none $\kappa$=100 seed 85 & 45 & 3 & 0 & $6\times10^{-5}$ & $2\times10^{-7}$ \\
(vi) & $p$=20 $n$=60 nonneg/none $\kappa$=10 seed 90 & 13 & 0 & 0 & $4\times10^{-11}$ & $2\times10^{-10}$ \\
(vi) & $p$=20 $n$=60 none/sumzero $\kappa$=10 seed 91 & 22 & 1 & 0 & $<10^{-13}$ & $1\times10^{-12}$ \\
(vi) & $p$=30 $n$=100 nonneg/none $\kappa$=100 seed 92 & 25 & 4 & 0 & $3\times10^{-11}$ & $-4\times10^{-12}$ \\
(vi) & $p$=30 $n$=100 nonneg/random $\kappa$=10 seed 93 & 12 & 0 & 0 & $3\times10^{-11}$ & $1\times10^{-11}$ \\
(vi) & $p$=50 $n$=150 nonneg/none $\kappa$=10 seed 94 & 23 & 0 & 0 & $2\times10^{-11}$ & $-1\times10^{-10}$ \\
(vi) & $p$=50 $n$=150 none/none $\kappa$=100 seed 95 & 107 & 28 & 1 & $<10^{-13}$ & $<10^{-13}$ \\
\hline
\end{tabular}}}
\end{table*}

\begin{table*}
\centering
\caption{The six wide ($n<p$) cases at five ridge strengths
$\epsilon/\max_i H_{ii}$. `Matched' is the maximum over $25$ log-spaced
$\lambda$ of the signed relative objective difference against the ridge-augmented
problem the run is solving; `unregularised' is the same signed difference against the
problem with $\epsilon=0$. 
The matched comparison is with a numerical solution of the same objective;
the unregularised comparison additionally includes the effect of ridge augmentation.
The latter generally increases with the ridge strength, although
neither its rate of increase nor the recorded point count is
uniform across cases. The largest matched difference over the $30$ runs
is $2\times10^{-9}$; over all sampled points of those runs the
differences in the opposite direction reach $-3.6\times10^{-7}$,
which a per-run maximum does not display.\label{tab:ridge}}
\smallskip
{\small
\CILFITWIDTH{%
\begin{tabular}{llrll@{\qquad}llrll}
\hline
Case & $\epsilon/\max_i H_{ii}$ & Points & Matched & Unreg. &
Case & $\epsilon/\max_i H_{ii}$ & Points & Matched & Unreg.\\
\hline
$p$=10 $n$=6 none/none seed 80 & $10^{-6}$ & 23 & $-5\times10^{-13}$ & $7\times10^{-8}$ & $p$=15 $n$=8 nonneg seed 83 & $10^{-6}$ & 16 & $1\times10^{-12}$ & $2\times10^{-7}$ \\
& $10^{-5}$ & 23 & $-4\times10^{-13}$ & $7\times10^{-6}$ & & $10^{-5}$ & 16 & $1\times10^{-12}$ & $2\times10^{-5}$ \\
& $10^{-4}$ & 23 & $-9\times10^{-13}$ & $3\times10^{-2}$ & & $10^{-4}$ & 12 & $6\times10^{-13}$ & $2\times10^{-3}$ \\
& $10^{-3}$ & 19 & $1\times10^{-13}$ & $3\times10^{-1}$ & & $10^{-3}$ & 10 & $2\times10^{-11}$ & $1\times10^{-1}$ \\
& $10^{-2}$ & 11 & $2\times10^{-14}$ & $3$ & & $10^{-2}$ & 12 & $1\times10^{-10}$ & $3$ \\
\addlinespace[2pt]
$p$=10 $n$=6 nonneg seed 81 & $10^{-6}$ & 8 & $4\times10^{-10}$ & $1\times10^{-8}$ & $p$=20 $n$=10 nonneg seed 84 & $10^{-6}$ & 6 & $1\times10^{-9}$ & $1\times10^{-9}$ \\
& $10^{-5}$ & 8 & $4\times10^{-10}$ & $1\times10^{-6}$ & & $10^{-5}$ & 6 & $1\times10^{-9}$ & $6\times10^{-8}$ \\
& $10^{-4}$ & 10 & $4\times10^{-10}$ & $1\times10^{-4}$ & & $10^{-4}$ & 6 & $1\times10^{-9}$ & $6\times10^{-6}$ \\
& $10^{-3}$ & 10 & $2\times10^{-9}$ & $8\times10^{-3}$ & & $10^{-3}$ & 6 & $2\times10^{-10}$ & $5\times10^{-4}$ \\
& $10^{-2}$ & 8 & $3\times10^{-14}$ & $1\times10^{-1}$ & & $10^{-2}$ & 6 & $3\times10^{-10}$ & $4\times10^{-2}$ \\
\addlinespace[2pt]
$p$=15 $n$=8 sumzero seed 82 & $10^{-6}$ & 27 & $-3\times10^{-15}$ & $3\times10^{-4}$ & $p$=20 $n$=10 box seed 85 & $10^{-6}$ & 45 & $6\times10^{-11}$ & $8\times10^{-9}$ \\
& $10^{-5}$ & 27 & $1\times10^{-14}$ & $1\times10^{-2}$ & & $10^{-5}$ & 45 & $8\times10^{-11}$ & $8\times10^{-7}$ \\
& $10^{-4}$ & 23 & $1\times10^{-13}$ & $4\times10^{-2}$ & & $10^{-4}$ & 45 & $2\times10^{-11}$ & $6\times10^{-5}$ \\
& $10^{-3}$ & 21 & $2\times10^{-14}$ & $2\times10^{-1}$ & & $10^{-3}$ & 44 & $2\times10^{-11}$ & $5\times10^{-3}$ \\
& $10^{-2}$ & 19 & $9\times10^{-15}$ & $7$ & & $10^{-2}$ & 31 & $2\times10^{-11}$ & $9\times10^{-1}$ \\
\hline
\end{tabular}}}
\end{table*}

\subsection{The per-kink certificate}\label{sec:tiers}

At each recorded point, we independently check primal
feasibility, dual feasibility, complementarity and
stationarity. The active-side stationarity residual is
$\lVert\bm{h}_a-\bm{H}_{aa}\bm{\beta}_a
-\bm{A}_{za}^{\rm t}\bm{\mu}_z
-\lambda\,\mathrm{sign}\,\bm{\beta}_a\rVert_\infty$,
normalised by $\lVert\bm{h}\rVert_\infty$.
For inactive coefficients we test the unclipped
$\lambda\bm{g}_{\bar a}$ against its bounds,
using the same normalisation. Testing this scaled
quantity avoids amplification of certificate error
by division by small $\lambda$.

All $39$ runs pass the specified checks at every
recorded point. The largest raw excess
$\max_j(|g_j|-1)_+$ is $2\times10^{-10}$;
the largest scaled subgradient-bound residual is
$3\times10^{-14}$, and the largest normalised
active-side stationarity residual is
$2\times10^{-14}$, against an acceptance tolerance
of $10^{-8}$.

\subsection{Verification between the kinks}\label{sec:segment_verification}

The preceding checks assess recorded points and a
sample of intermediate values. Since the output is
affine on each recorded segment, one can also verify
its certificate throughout the interval without
sampling. Stationarity, primal feasibility,
multiplier signs, and the identities fixing either
a constraint slack or its multiplier to zero are
affine in $\lambda$, so endpoint bounds extend
over the segment.

Complementarity requires additional care. A product
of affine functions may vanish at both endpoints
without vanishing in between. We therefore verify
the appropriate zero factor on the whole interval,
rather than testing only the product at its ends.
The subgradient inequalities are checked for every
inactive coefficient, including those in $b$.
Boundary-set membership does not exempt a coordinate
from $|\lambda g_j|\le\lambda$.

The checker also locates possible interior events:
an active coefficient crossing zero, an inactive
subgradient reaching either boundary, a slack row
becoming tight, or a positive multiplier reaching
zero. The subgradient test includes the held-coordinate
crossing denoted by \textbf{E5}. Thus the four
types of residual tests cover the five primal event
labels. A change of certificate face, \textbf{E6},
does not bend the primal path, and the point at
which it occurs is recorded only if the primal
direction changes there or \textbf{E1} or \textbf{E3}
occurs.

These tests use the coefficients and multipliers
stored at the two segment endpoints and their
affine interpolation; they do not recover a
replacement certificate. They therefore assess the
recorded output itself. In floating point the
results are residual bounds, with the stated
verification tolerances, rather than exact proofs.

Two properties should be distinguished. The
checker verifies that the compressed output carries
a valid KKT certificate: on each recorded segment,
the affine interpolation of the certificates
recorded at its two ends satisfies every condition
above. It does not verify that the traversal
followed the canonical certificate through the
steps removed from that segment. At such a step
the canonical certificate can change slope while
the primal path stays affine, so the interpolated
certificate need not coincide with the canonical
one in between; it remains a valid certificate for
the whole segment whenever the checks pass. A
necessary invariant of certificate following is
checked instead, at the end of every step, removed
or not: the implementation recomputes the
canonical certificate there and compares it with
the one propagated along the lexicographic rate,
which must coincide by
Proposition~\ref{prop:efficient_step}(iii). The
invariant does not by itself show that each step
exhausts its cell, since steps cut short along the
correct certificate would also satisfy it. Over
the $45$ paths the largest difference, normalised
by $\lVert\bm h\rVert_\infty$, is $3.5\times10^{-11}$,
and no step substituted a fallback for the rule's
own solves (Section~\ref{sec:follow}).
The checker is distributed as
\texttt{classo.verify\_segments}.

The $45$ paths are the $39$ synthetic cases together with the
six diabetes configurations of Section~\ref{sec:diabetes_app};
their $625$ recorded points define $580$ segments,
from within which $141$ points that are not kinks were removed.
Over those segments,
the largest normalised residuals are $9\times10^{-15}$ for
stationarity, $5\times10^{-15}$ for primal
feasibility, $9\times10^{-18}$ for multiplier signs,
$9\times10^{-15}$ for the zero-slack identity on
strongly active rows, and $8\times10^{-15}$ for
zero multipliers off the binding set. No interior
event was missed under the verification tests.
The subgradient condition was tested $13\,622$ times,
once for each inactive coordinate at each end of each
segment and once at each interior breakpoint of
$|\lambda g_j|-\lambda$. Its largest
positive excess, normalised by $\lVert\bm{h}\rVert_\infty$,
was $7.0\times10^{-11}$, below the tolerance
of $10^{-8}$.

The dual-recovery solve is performed to a nominal tolerance
of $10^{-12}$, and the subgradient acceptance tolerance
is $10^{-8}$, specified independently of the active-side
tolerance. If the tighter solve fails, the default solve
is attempted, but its result must still pass the same
residual tests. No such fallback occurred on the reported
benchmark. These acceptance conditions are distinct from
the solver's own stopping criterion.

The strong/weak partition is consistent with the recovered
multipliers: multipliers on rows classified as weakly active
are at most $2\times10^{-13}$, against a classification
threshold of $10^{-10}$, and no row declared strongly active
becomes slack along a segment without a release event. The
accuracy of certificate recovery must support the thresholds
subsequently used to classify its components.

\subsection{Agreement with an exact implementation}\label{sec:prototype}

Theorem~\ref{thm:termination} concerns exact arithmetic. We
compare the floating-point implementation with an independent
implementation of the rule in rational arithmetic, which enumerates the independent active subsets of the
parametric problems, carries a minimum-norm certificate on each
outgoing segment, and uses no tolerances. Its cost restricts the
comparison to small problems. These have integer data: an
upper-triangular design with diagonal entries $1$ or $2$ and
off-diagonal entries in $\{-2,\dots,2\}$, an integer response, and
between one and $p-1$ random constraint rows mixing equalities and
inequalities, with the origin feasible; a duplicated row and a pair
of opposed rows that tie are inserted at random. For each problem the
recorded points of CLASSO are compared, in $\lambda$ and
$\bm{\beta}$, with the changes of primal direction found by the exact
implementation.

For each $p$ from $2$ to $6$ we draw $40$ problems from a fixed seed,
and $10$ more at $p=7$, and run the exact implementation on every one
to completion, with no time limit. Its cost is heavy-tailed: most
problems take seconds or minutes, but the slowest at $p=6$ took
$66$ minutes and the slowest at $p=7$ $2.5$ hours. On all $210$
problems the CLASSO path is certified, records the same number of
points, and its values of $\lambda$ and $\bm{\beta}$ agree with the
exact ones to at most $4.6\times10^{-13}$.

Because the primal path is unique, this agreement shows that the
recorded output is exact to rounding on these deliberately degenerate
problems; it does not by itself show that the same certificate rule
was followed. The step-end check of
Section~\ref{sec:segment_verification}, a necessary invariant of
that rule, also holds: on all $210$ problems, the propagated
certificate agrees with the
recomputed canonical certificate at the end of every step to
$1.6\times10^{-9}$, normalised by $\lVert\bm h\rVert_\infty$,
and no fallback occurs.

The comparisons concern small, integer problems, chosen to include
equalities, duplicated and opposed rows, and ties. They are evidence
that the implementation follows the rule of
Theorem~\ref{thm:termination}, not a proof that it does so on every
input.

\subsection{Continuous cross-validation}
\label{sec:cv}

An application of the affine path representation is
continuous optimisation of a cross-validation criterion.
Within a segment, the held-out predictions
$\bm{X}_v\bm{\beta}(\lambda)$ are affine, so their
squared loss is quadratic in $\lambda$ with a
non-negative leading coefficient. Its constrained
minimum is obtained directly by checking the
stationary point, when present, and the segment
endpoints. For $K$-fold cross-validation, one must
compute a path for each training fold. The sum of
the validation losses is piecewise quadratic on
the union of the fold breakpoint sets, and its
global minimum follows by minimising over those
intervals. The full-data path is then used to
obtain the final fitted model.

On the diabetes data, five-fold cross-validation
(\texttt{reproduction/cv\_exact.py}) gives a minimum
loss only $10^{-7}$ to $5\times10^{-7}$ lower,
in relative terms, than a $100$-point grid search.
The practical difference in the criterion value
is therefore negligible. Its location can still
affect the selected support: for
\textsc{none/none}, the continuous minimum at
$\lambda=130.40$ and the grid minimum at
$129.80$ lie on opposite sides of a breakpoint,
giving four and five selected coefficients,
respectively. The supports agree for the other
three tested configurations.

This example demonstrates sensitivity of support
selection near a flat validation optimum; it
does not establish that the continuously selected
model is statistically preferable. A finer grid
or adaptive search could also locate the minimum.
The computational advantage of the path
representation is that the criterion can be
minimised directly over its known quadratic
pieces once the fold paths have been computed.

\subsection{Scaling to larger \texorpdfstring{$p$}{p}}\label{sec:scaling}

For the scaling study we generate problems with
$n=4p$, a Gaussian design of condition number
$10$, and a sparse generating coefficient vector
with five non-zero elements. We consider no
constraints, a sum-non-negativity inequality,
a sum-to-zero equality, and the mixed
\textsc{sumnonneg/betadiff} configuration defined
in Section~\ref{sec:diabetes_data}.
Table~\ref{tab:scalability} reports point counts,
LP/QP solve counts and wall-clock times for
$p=20$ to $5000$.

\begin{table}
\centering
\caption{CLASSO scaling on synthetic problems at $n=4p$, condition
number $10$, $k=5$ true nonzeros. `Points' is the number of recorded
path points, including endpoints; `Solves' the total number of LP and QP solves. Every run passes
the per-kink certificate with no rejection. Wall-clock times are for
the path computation alone, measured on exclusively reserved nodes of
an Intel Xeon (Ice Lake) cluster with eight threads per run and
FlexiBLAS, one configuration per node. Peak resident memory is
$2.1$\,GB at $p=2000$ and $11.5$\,GB at $p=5000$. Point and solve
counts were identical in every environment in which they were
repeated; those with $p\le1000$ were reproduced exactly by the cited
release on a different machine and software stack, whose code
differs from that of these runs only in its handling of the final
step. The four $p=5000$ runs, alone among the table, shared a node,
with each other and with other runs, so their times carry a
contention penalty and are excluded from the cost fit quoted in the
text; the reported point and solve counts are recorded separately from the timings.%
\label{tab:scalability}}
\smallskip
\begin{tabular}{lrrrrr}
\hline
Configuration & $p$ & $n$ & Points & Solves & Time (s) \\
\hline
\textsc{none /}             &    20 &     80 &    21 &    23 &    0.021 \\
\textsc{none}               &    50 &    200 &    63 &    65 &    0.054 \\
                            &   100 &    400 &   111 &   113 &    0.169 \\
                            &   200 &    800 &   235 &   237 &    0.988 \\
                            &   500 &   2000 &   587 &   589 &     12.6 \\
                            &  1000 &   4000 &  1207 &  1209 &     97.3 \\
                            &  2000 &   8000 &  2325 &  2327 &    827.8 \\
                            &  5000 &  20000 &  5835 &  5837 &    30418 \\
\hline
\textsc{sumnonneg /}        &    20 &     80 &    25 &    27 &    0.026 \\
\textsc{none}               &    50 &    200 &    53 &    55 &    0.048 \\
                            &   100 &    400 &   120 &   124 &    0.236 \\
                            &   200 &    800 &   227 &   229 &    0.951 \\
                            &   500 &   2000 &   591 &   593 &     12.8 \\
                            &  1000 &   4000 &  1167 &  1169 &     96.1 \\
                            &  2000 &   8000 &  2337 &  2339 &    807.5 \\
                            &  5000 &  20000 &  5832 &  5834 &    30013 \\
\hline
\textsc{none /}             &    20 &     80 &    30 &    37 &    0.067 \\
\textsc{sumzero}            &    50 &    200 &    56 &    61 &    0.114 \\
                            &   100 &    400 &   128 &   144 &    0.453 \\
                            &   200 &    800 &   236 &   256 &    1.890 \\
                            &   500 &   2000 &   582 &   625 &     23.7 \\
                            &  1000 &   4000 &  1180 &  1272 &    178.5 \\
                            &  2000 &   8000 &  2386 &  2581 &     1685 \\
                            &  5000 &  20000 &  5902 &  6355 &    52570 \\
\hline
\textsc{sumnonneg /}        &    20 &     80 &    22 &    30 &    0.059 \\
\textsc{betadiff}           &    50 &    200 &    60 &    74 &    0.136 \\
                            &   100 &    400 &   116 &   183 &    0.462 \\
                            &   200 &    800 &   244 &   446 &    2.316 \\
                            &   500 &   2000 &   579 &   761 &     23.3 \\
                            &  1000 &   4000 &  1177 &  1821 &    189.2 \\
                            &  2000 &   8000 &  2366 &  2776 &     1659 \\
                            &  5000 &  20000 &  5883 &  8584 &    43681 \\
\hline
\end{tabular}
\end{table}

The recorded point counts are approximately linear
in $p$ for all four tested configurations.
Their ratio to $p$ lies between $1.05$ and
$1.50$ across the table, narrowing to
$1.16$--$1.21$ for $p\ge500$ and
$1.17$--$1.18$ at $p=5000$.
Thus these constraint configurations give similar
counts at larger dimensions, although they
change the locations and identities of the
events. This empirical behaviour is much more
favourable than the combinatorial worst-case
bound and should not be interpreted as a bound
for general constraint geometries.

Without constraints, and with the single
inequality, the solve count exceeds the recorded
point count by the two initialisation solves
(four in one row), since the direction requires one QP per
iteration and the refinements recover other
quantities through linear algebra. With an
equality constraint, the certificate's rate and
horizon programmes, and steps that change only
the certificate and are removed from the output
(Section~\ref{sec:follow}), add further solves:
$7$--$23$ per cent for \textsc{none/sumzero}
and up to $83$ per cent, at $p=200$, for
\textsc{sumnonneg/betadiff}. Dense
factorisations in the directional and fixed-state
systems have nominal cost $O(|\alpha|^3)$.
Combining this with a number of iterations
proportional to $p$ permits quartic total
scaling when the active blocks are of order $p$.
Over the measured range $p=500$ to $2000$,
however, fitted exponents are $2.99$--$3.08$.
This finite-range fit describes the observed
costs; it does not determine the asymptotic
complexity.

At $p=2000$, each run takes approximately
$13$--$28$ minutes on a reserved node.
The $p=5000$ runs take $8.3$--$14.6$ hours,
or $1.6$--$2.4$ times the extrapolated fit.
Since those four runs shared a node, the
difference cannot be separated reliably into
resource contention and a change in scaling,
and they are excluded from the fit.
Full rank and conditioning decompositions
contribute substantially to the runtime.
Updating these decompositions and warm-starting
the directional QP are possible improvements.

\section{Application to the diabetes dataset}\label{sec:diabetes_app}

We next apply CLASSO to the diabetes data used in
the original LARS study \citep{Efron2004}.
After specifying the constraint configurations,
we examine active-set margins, compare
computational costs with a discrete sweep, and
analyse the local direction system of
\citet{Gaines2018} on the same paths.

\subsection{The diabetes dataset and constraint configurations}\label{sec:diabetes_data}

The diabetes dataset of \citet{Efron2004} contains $n=442$ diabetes
patients, each described by $p=10$ baseline covariates (age, sex,
body mass index, mean arterial blood pressure, and six blood serum
measurements), together with a scalar response variable representing
a quantitative measure of disease progression one year after
baseline. The data are standardised so that each column of
$\bm{X}$ has zero mean and unit $\ell_2$ norm. This dataset is the
original LARS benchmark of \citet{Efron2004} and is widely used in
the path-following literature; we adopt it here for
comparability.

On the diabetes dataset we use six constraint configurations. Each
configuration is specified as a pair \textsc{inequality}/\textsc{equality},
where the inequality options are \textsc{none} (unconstrained),
\textsc{nonneg}~($\bm{\beta}\ge\bm{0}$),
\textsc{nonpos}~($\bm{\beta}\le\bm{0}$), and
\textsc{sumnonneg}~($\bm{1}^{\rm t}\bm{\beta}\ge 0$); and the equality
options are \textsc{none}, \textsc{sumzero}~($\bm{1}^{\rm t}\bm{\beta}=0$),
and \textsc{betadiff}~($\beta_3-\beta_9=50$, chosen as a
non-homogeneous but feasible two-coefficient constraint). The six
configurations are
\textsc{none/none} (the unconstrained LARS reference),
\textsc{nonneg/none}, \textsc{nonpos/none}, \textsc{none/sumzero},
\textsc{nonneg/betadiff}, and \textsc{sumnonneg/betadiff}. Together
they span pure inequalities (sign constraints of varying sparsity), a
pure equality (\textsc{sumzero}), and two mixed configurations with
non-homogeneous equalities.

These constraints are chosen to test distinct features
of the algorithm, rather than as six scientific
models of diabetes progression. Individual sign
bounds exercise one-sided coefficient entries;
the sum inequality tests an aggregate restriction;
and the equalities couple coefficients, with
$\beta_3-\beta_9=50$ providing a non-homogeneous
example. The value $50$ makes that equality
nontrivial relative to the unconstrained fit.
Non-negativity and non-positivity test both
orientations of the bound handling, but their
paths on the same data need not be symmetric.
Likewise, $\bm{1}^{\rm t}\bm{\beta}\ge0$
restricts the coefficient sum and does not imply
non-negative predictions for centred covariates.

\subsection{Empirical margins of the nondegeneracy conditions}\label{sec:nondegeneracy}

Although the structural results do not assume
\textbf{(B2)}--\textbf{(B4)}, numerical distances
from classification boundaries help assess the
sensitivity of the implementation.
Table~\ref{tab:nondegeneracy} reports, at
certified points with $\lambda>0$, the smallest
inactive subgradient margin outside $b$,
the smallest positive multiplier on $s$,
the smallest slack on $\bar z$, and the
largest boundary set encountered.
Coordinates held by one-sided bounds are
treated through their bound multipliers.

The positive multiplier and slack minima are
conditional on the declared sets, so they
measure separation of those classifications
rather than independently proving strict
complementarity over every open segment.
Similarly, a boundary coordinate or weakly
active row may occur at an event even when
the adjacent open segments are nondegenerate.

\begin{table}
\centering
\caption{Nondegeneracy margins along the six diabetes paths, from the
numerically recovered certificates at recorded points with $\lambda>0$.
$\delta_{\min}^{(\bar a\setminus b)}$ is the minimum of $1-|g_j|$
over inactive coordinates neither in the boundary set nor held at
zero by a binding sign bound (a held coordinate's margin is its
bound's multiplier, Section~\ref{sec:one_sided});
$\mu_{\min}^{(s)}$ the minimum multiplier over strongly active rows;
$\mbox{slack}_{\min}^{(\bar z)}$ the minimum slack over non-binding
rows; $|b|_{\max}$ the largest boundary set at any kink. `---': the set
is empty throughout.\label{tab:nondegeneracy}}
\smallskip
\small\setlength{\tabcolsep}{4pt}
\begin{tabular}{lrrrrr}
\hline
Constraints & Points & $\delta_{\min}$ & $\mu_{\min}$ & $\mbox{slack}_{\min}$ & $|b|_{\max}$\\
\hline
\textsc{none / none}           & 12 & $1.3\times10^{-3}$ & --- & --- & 1 \\
\textsc{nonneg / none}         &  5 & --- & $30$ & $32$ & 4 \\
\textsc{nonpos / none}         &  2 & --- & $56$ & $514$ & 9 \\
\textsc{none / sumzero}        &  9 & $2.2\times10^{-3}$ & --- & --- & 2 \\
\textsc{nonneg / betadiff}     &  4 & --- & $33$ & $29$ & 4 \\
\textsc{sumnonneg / betadiff}  &  9 & $1.7\times10^{-2}$ & --- & $50$ & 1 \\
\hline
\end{tabular}
\end{table}

On the three paths with free inactive coordinates,
the reported subgradient margins outside $b$
are at least $10^{-3}$. For the sign-constrained
paths, zero coefficients are held by bound rows,
and the corresponding positive multipliers and
slacks are of order $10$--$10^2$. On these paths
the canonical certificate divides each held
coordinate's stationarity combination at minimum
norm (Section~\ref{sec:one_sided}), and the
resulting subgradient can reach $\pm1$; such
coordinates are then in $b$, which accounts for
the larger boundary sets in those rows.

At the start of \textsc{none/sumzero}, the
sum-to-zero multiplier at $\lambda_{\max}$
equalises the largest and smallest components
of $\bm{h}$, so $|g_2|=|g_6|=1$ exactly and both
coordinates enter $b$; the next largest is
$|g_8|=0.958$. This tie is forced by the
constraint rather than numerically close, and
a computed certificate resolves it only to
rounding, which is one reason for choosing a
boundary threshold consistent with certificate
accuracy (Appendix~\ref{sec:safeguards}). On
the other two paths with free coordinates,
\textsc{none/none} and \textsc{sumnonneg/betadiff},
the boundary set has one element.

\subsection{Benchmarking against a discrete $\lambda$-sweep baseline}\label{sec:benchmarking}

We compare CLASSO with a $1000$-point logarithmic
$\lambda$-sweep, solving the full constrained LASSO
as a $2p$-dimensional QP at each grid value.
The comparison uses the six diabetes configurations
described above and measures the cost of these
two implementations on those problems.

For each configuration we record LP and QP call counts
and wall-clock time, using the same optimisation packages
for both methods.\footnote{The
implementation uses ProxQP for QPs and SciPy's \texttt{linprog}
(HiGHS) for LPs. Timings are indicative rather than highly
optimised, but the relative numbers are informative. Two environments
were used. The stress suite, the diabetes results of this section and
the comparisons below were computed with \textsc{proxsuite} $0.7.1$,
\textsc{qpsolvers} $4.4.0$, NumPy $1.26.4$ and SciPy $1.15.1$ on a
laptop; the scaling study of Table~\ref{tab:scalability} was computed
on cluster nodes with \textsc{proxsuite} $0.7.3$, \textsc{qpsolvers}
$4.13.0$ and SciPy $1.13.1$. Point and solve counts agree between the two environments for the same
code; only wall-clock times differ, and no table mixes them.}

Using ProxQP for both methods reduces differences due
to the underlying solver, although their subproblem
sizes, conditioning and recovery operations differ.
We therefore report both solve counts and wall-clock
times. A count of LP and QP calls does not include
the cost of the linear factorisations used in
CLASSO's refinement and certificate stages.\footnote{The
\texttt{genlasso} package \citep{Arnold2016} traces
a dual path for penalties
$\lVert\bm{D}\bm{\beta}\rVert_1$. Encoding the present
hard constraints as large penalty rows would change
the formulation being compared.}

The sweep uses $N=1000$ logarithmically spaced
values over $\lambda\in[10^{-3},10^3]$, with
a multiplicative change of approximately
$\exp(-0.014)$ per step. For context, the
usual grid-based LASSO interfaces in
\texttt{glmnet} and \texttt{scikit-learn} use
approximately $100$ values by default, although
their coordinate-descent solvers differ from
the general constrained QPs used here.

Grid density determines how accurately interpolation
resolves a breakpoint lying between adjacent samples.
On an interval containing no breakpoint, the exact
primal path is affine; an interval containing a
change of direction incurs an interpolation error
depending on the local spacing and slopes.
A meaningful accuracy--cost comparison therefore
requires measuring achieved path error across
several grid densities. The present experiment
instead fixes the grid and reports the resulting
solve counts and timings.

\begin{table*}
\centering
\caption{Local-solve counts and wall-clock timings for CLASSO
path-following versus the $1000$-point $\lambda$-sweep on the diabetes
dataset. `Points' counts recorded path points, including endpoints.
`Solves' is the total number of LP and QP solves. `Ratio' is
the sweep solve count divided by the CLASSO solve count; `Speedup'
is the corresponding wall-clock time
ratio.\label{tab:benchmark_diabetes}}
\smallskip
\CILFITWIDTH{%
\begin{tabular}{lcrrrrcrrcc}
\hline
Constraints & Points & \multicolumn{4}{c}{CLASSO} & &
\multicolumn{2}{c}{Sweep} & Solve & Speedup\\
\cline{3-6}\cline{8-9}
            &       & LPs & QPs & Total & Time (s) & & Solves & Time (s) &
ratio &     \\
\hline
\textsc{none / none}          & 13 &  2 & 13 &  15 & 0.008 & &   1002 & 0.156 &  67$\times$ &  20$\times$\\
\textsc{nonneg / none}        &  6 & 22 & 11 &  33 & 0.041 & &   1002 & 0.177 &  30$\times$ & 4.3$\times$\\
\textsc{nonpos / none}        &  3 & 14 &  7 &  21 & 0.030 & &   1096 & 0.189 &  52$\times$ & 6.4$\times$\\
\textsc{none / sumzero}       & 10 &  2 & 10 &  12 & 0.010 & &   1002 & 0.129 &  84$\times$ &  13$\times$\\
\textsc{nonneg / betadiff}    &  5 & 20 & 10 &  30 & 0.037 & &   1002 & 0.182 &  33$\times$ & 4.9$\times$\\
\textsc{sumnonneg / betadiff} & 10 &  2 & 10 &  12 & 0.011 & &   1002 & 0.155 &  84$\times$ &  15$\times$\\
\hline
Total                & 47 & 62 & 61 & 123 & 0.14  & &   6106 & 0.99  &  50$\times$ & 7.2$\times$\\
\hline
\end{tabular}}
\end{table*}

Table~\ref{tab:benchmark_diabetes} reports the
number of LP and QP calls, including initialisation,
and the total path-computation time.
The sweep counts include two endpoint solves in
addition to the grid.\footnote{For
\textsc{nonpos/none}, $94$ grid values required
a second solver attempt using warm-start
information, giving $1096$ calls instead of
$1002$.}
Across the six cases, CLASSO uses $123$ solves
and the sweep $6106$, a ratio of $50$.
The corresponding times are $0.14$\,s and
$0.99$\,s, giving a speedup of approximately
$7$. Individual solve-count ratios range
from $30$ to $84$, and time ratios from
$4.3$ to $20$. The three configurations with
individual sign bounds account for most of
the CLASSO solves: on them the certificate's
rate and horizon programmes add two LPs per
step, and the certificate-only steps of
Section~\ref{sec:follow} add QPs that do not
appear as recorded points.

The smaller time ratio reflects the different
costs per call and the additional linear algebra
performed by CLASSO. Its dense factorisations,
rank and conditioning checks, and lexicographic
initialisation contribute work that is not
represented by the LP/QP call count. The
sweep repeatedly solves a QP with the same
structure and a changed linear term.

These figures apply to the specified
$1000$-point baseline. A coarser grid would
cost less, while a finer grid would resolve
breakpoints more closely. We do not infer a
general speedup at matched path accuracy,
or relative to other parametric-QP methods,
from this experiment. CLASSO additionally
returns the segment boundaries directly,
which can be useful when their locations
are required.

\begin{figure*}
\centering
\includegraphics[width=\textwidth]{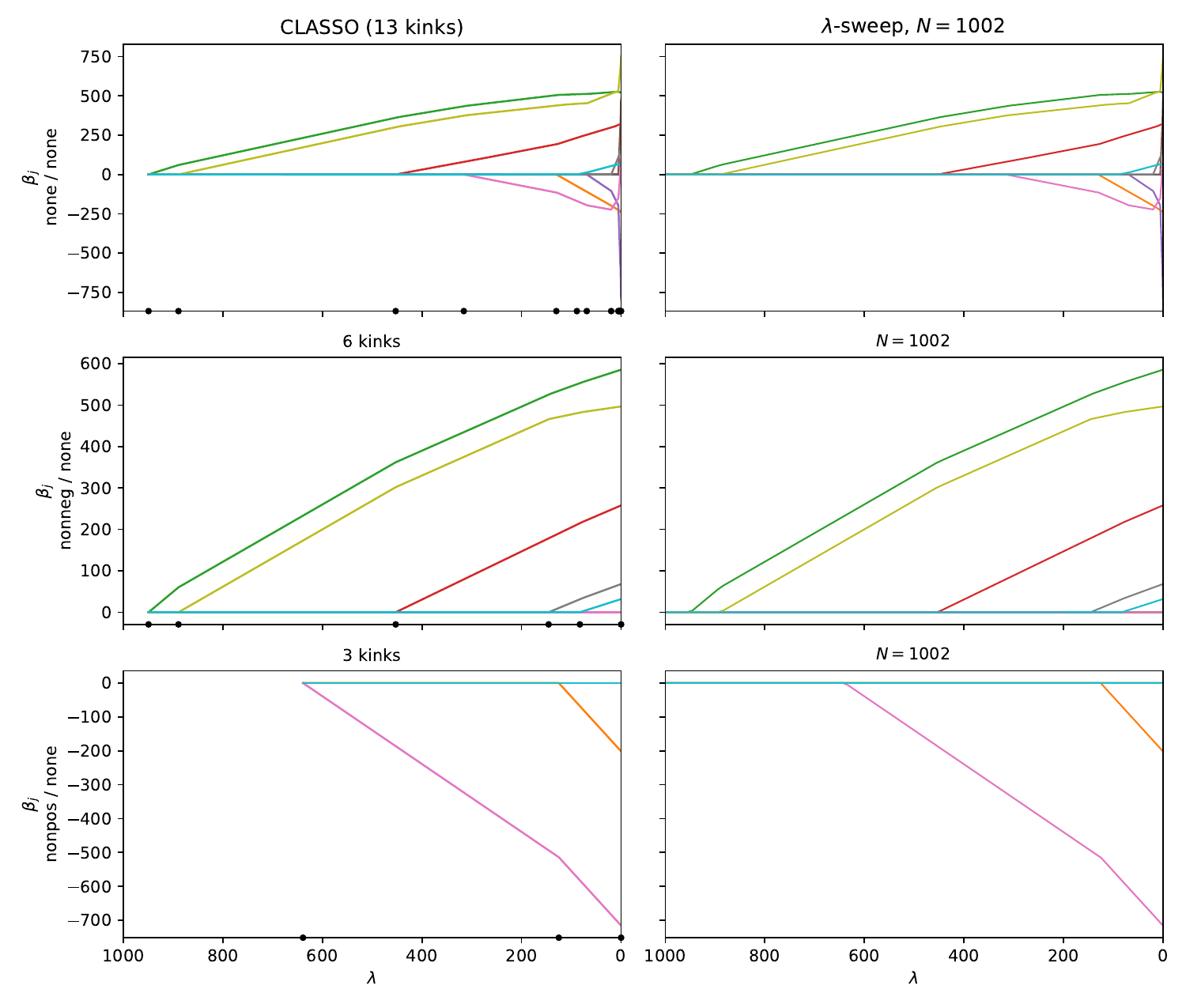}
\caption{Coefficient paths for the three diabetes configurations
without equality constraints (rows, top to bottom:
\textsc{none}/\textsc{none}, \textsc{nonneg}/\textsc{none},
\textsc{nonpos}/\textsc{none}). Left column: CLASSO piecewise-linear
path; black markers on the $\lambda$-axis indicate the kink
locations returned by the path-following algorithm (some points
overlap closely, so are difficult to see). Right column:
$1000$-point $\lambda$-sweep baseline. Line colours identify the ten
diabetes coefficients consistently across panels. The $\lambda$ axis
is linear with $\lambda$ decreasing to the right.\label{fig:diabetes_eqnone}}
\end{figure*}

\begin{figure*}
\centering
\includegraphics[width=\textwidth]{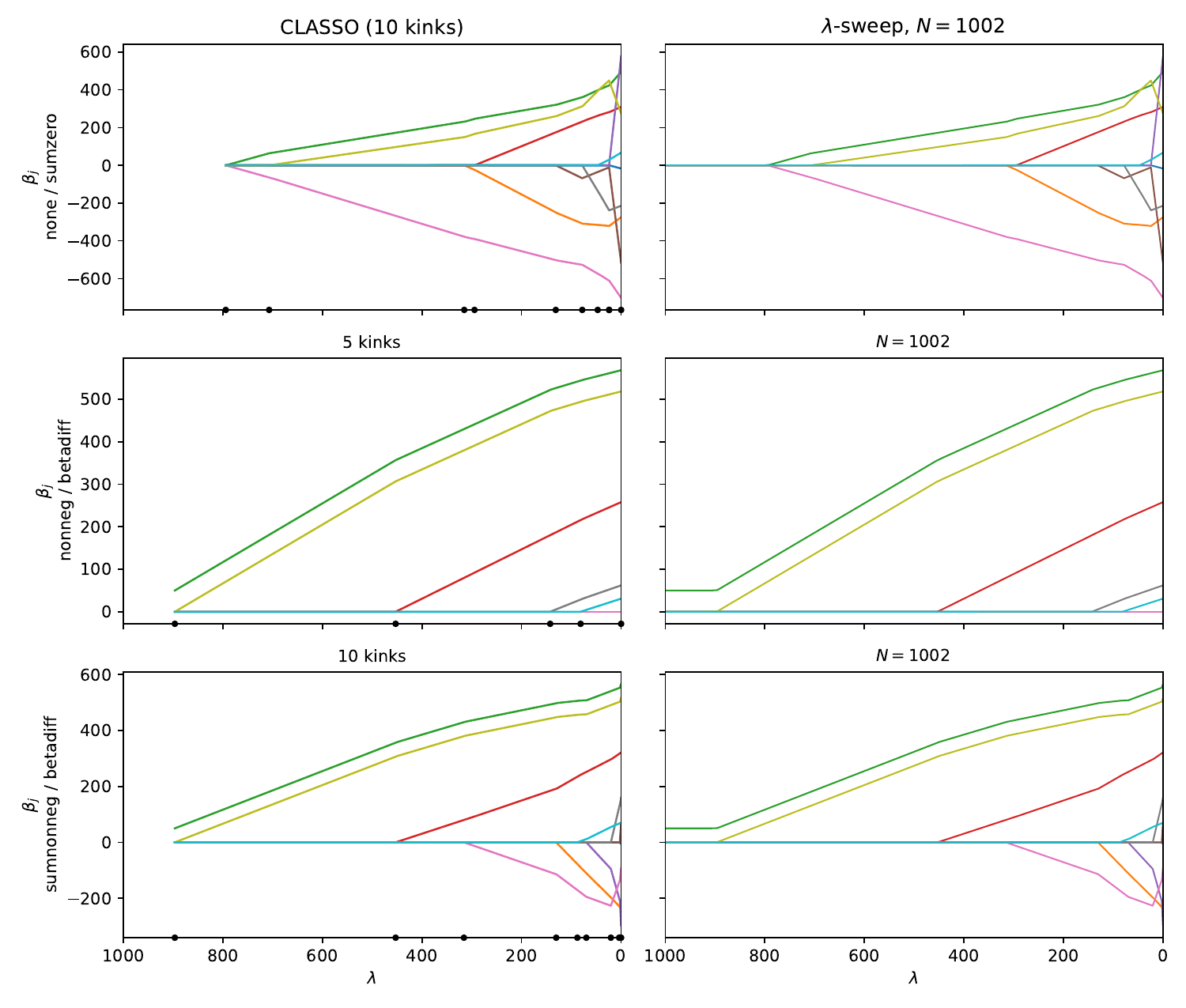}
\caption{Coefficient paths for the three diabetes configurations
with an equality constraint (rows, top to bottom:
\textsc{none}/\textsc{sumzero}, \textsc{nonneg}/\textsc{betadiff},
\textsc{sumnonneg}/\textsc{betadiff}). Layout as in
Figure~\ref{fig:diabetes_eqnone}.\label{fig:diabetes_eqcon}}
\end{figure*}

Figures~\ref{fig:diabetes_eqnone}
and~\ref{fig:diabetes_eqcon} compare the coefficient
paths for all six configurations. The two
methods are visually indistinguishable at the
figure resolution; the numerical residual
tests above provide the more stringent
assessment of agreement.

For \textsc{none/none}, the CLASSO output
agrees with the LASSO path returned by
\texttt{scikit-learn}'s \texttt{lars\_path},
using \texttt{method='lasso'}. Both return
$13$ recorded points, including the endpoints,
and the $\lambda$ values agree to the six
decimal places reported. The final coefficient
vectors differ by $1.6\times10^{-11}$ in
$\lVert\cdot\rVert_\infty$; the CLASSO endpoint
differs from ordinary least squares by
$3.8\times10^{-12}$.

The event near $\lambda=2.18$ includes the
entry of $\beta_0$ and departure of $\beta_6$.
The latter subsequently re-enters with the
opposite sign at $\lambda=1.310441$.
In CLASSO this is represented by \textbf{E5},
while \texttt{lars\_path} records an ordinary
entry. The comparison therefore tests the event
representation as well as the primal direction.

The constrained paths illustrate the expected
changes in coefficient behaviour. Individual
sign bounds remove branches with the forbidden
sign, while a sum-to-zero equality couples
positive and negative contributions. The
non-homogeneous equality in
\textsc{betadiff} maintains the prescribed
offset between its two coefficients.
These examples demonstrate the range of
constraint interactions accommodated by the
directional formulation.

\subsection{Comparison with the linear-system approach of Gaines et al.}\label{sec:gaines_comparison}

We finally compare the local direction calculation
with the linear system of \citet{Gaines2018}.
This is an examination of the active/binding
system at points on the computed path, rather
than an end-to-end comparison with an
independent implementation of their algorithm.
The matrix is
\begin{equation}
\bm{G}
=
\left(\begin{array}{cc}
\bm{H}_{aa} & \bm{A}_{za}^{\rm t} \\
\bm{A}_{za} & \bm{0}
\end{array}\right).
\label{eq:gainesmatrix}
\end{equation}

Proposition~\ref{prop:gaines_invertibility}
relates invertibility to the rank of
$\bm{A}_{za}$. Table~\ref{tab:gaines_comparison}
reports the points at which the matrix is
singular and separates the mutually exclusive
cases $|z|>|a|$ and row dependence with
$|z|\le|a|$. The boundary-set column records
a separate issue: coordinates at a subgradient
boundary may need admission to the free block,
but their presence alone is not a singularity
criterion.

A singular matrix need not make the equations
inconsistent. We therefore repeated the local
calculation using a least-squares solve and
the augmented free block $\alpha=a\cup b$,
as implemented in
\texttt{reproduction/gaines\_rank\_aware.py}.
The published implementation of
\citet{Gaines2018}, \texttt{ConstrainedLasso.jl},
likewise falls back to a pseudo-inverse, the
minimum-norm least-squares solution, when its
direction system is singular, so this is the
calculation it performs in that case.
For all six configurations, the resulting
systems were consistent to rounding precision.
Where the binding block did not pin a
candidate entering coordinate, the computed
direction agreed with the traversed segment
to relative differences between $10^{-15}$
and $10^{-9}$.

On the three sign-constrained configurations,
however, an entering coordinate's binding
sign row remains in $z$. Imposing that row
as an equality forces its direction to zero,
even after the free block is augmented.
The relative direction error is then
$0.69$--$1.0$. Thus a rank-aware solve
resolves consistency of the linear equations
but does not determine which weakly active
rows should be released. The directional QP
makes this decision through its inequality
constraints.

There is a separate choice concerning which
boundary coordinates move. On
\textsc{none/none} near $\lambda=2.18$,
the directional QP holds $\beta_6$ at zero,
whereas admitting every boundary coordinate
to an unrestricted free block gives a
relative direction error of $2.5$.
This is an ablation of the admission rule:
\citet{Gaines2018} uses a conditional test,
so the discrepancy should not be attributed
to their complete algorithm.

Finally, singularity leaves multiplier rates
undetermined even when the primal direction
is unique. For a positive-definite primal
block, the nullspace has the form
$\{\bm{0}\}\times\ker\bm{A}_{z\alpha}^{\rm t}$.
In the tested singular systems its dimension
is between $5$ and $9$. Changing the multiplier
rates within this freedom can change an
inactive subgradient rate, by up to $1$
in the reported normalisation, and hence
change the predicted certificate event.
The primal path itself remains unique.
The multiplier rates therefore need a
selection rule beyond the linear system;
CLASSO uses the lexicographic rate of
Section~\ref{sec:follow}.

\begin{table}
\centering
\caption{Comparison of local direction computation on the diabetes
dataset ($n=442$, $p=10$) under six constraint configurations.
Each configuration is specified as inequality\,/\,equality, where the
inequality options are: \textsc{none} (unconstrained), \textsc{nonneg}
($\bm{\beta}\ge\bm{0}$), \textsc{nonpos} ($\bm{\beta}\le\bm{0}$), and
\textsc{sumnonneg} ($\bm{1}^{\rm t}\bm{\beta}\ge0$); and the equality options
are: \textsc{none}, \textsc{sumzero} ($\bm{1}^{\rm t}\bm{\beta}=0$), and \textsc{betadiff}
($\beta_3-\beta_9=50$). `Points' counts certified path points with $\lambda>0$, excluding the final endpoint. `Gaines
singular' is the number of them at which the matrix $\bm{G}$ in
(\ref{eq:gainesmatrix}) is singular. Singularity mechanisms:
\textit{(i)}~$|z|>|a|$; \textit{(ii)}~dependent rows in
$\bm{A}_{za}$ with $|z|\le|a|$. Column \textit{(iii)} counts kinks
with boundary variables present (requiring the Gaines workaround but
not itself causing singularity).\label{tab:gaines_comparison}}
\smallskip
\CILFITWIDTH{%
\begin{tabular}{lccccc}
\hline
Constraints & Points & Gaines & \multicolumn{2}{c}{Singularity} & Boundary \\
            &       & singular  & \textit{(i)} & \textit{(ii)} & \textit{(iii)} \\
\hline
\textsc{none / none}           & 12 &  0 &  0 &  0 & 12 \\
\textsc{nonneg / none}         &  5 &  5 &  5 &  0 &  5 \\
\textsc{nonpos / none}         &  2 &  2 &  2 &  0 &  2 \\
\textsc{none / sumzero}        &  9 &  1 &  1 &  0 &  9 \\
\textsc{nonneg / betadiff}     &  4 &  4 &  4 &  0 &  4 \\
\textsc{sumnonneg / betadiff}  &  9 &  0 &  0 &  0 &  9 \\
\hline
\end{tabular}}
\end{table}

For \textsc{none/none} and
\textsc{sumnonneg/betadiff}, the matrix is
nonsingular at every tested point. The
\textsc{none/sumzero} system is singular
only at initialisation, where the active
support is empty. Individual sign constraints
give the more persistent failure: a binding
row is supported on an inactive coordinate,
so its restriction to the active columns is
zero. For these data, fewer than half the
coefficients are active at the tested
sign-constrained points, and the failures
fall into case \textit{(i)}.

For example, at kink~1 of
\textsc{nonneg/none}, $|a|=1$ and $|z|=9$.
The matrix $\bm{G}$ is $10\times10$,
but $\bm{A}_{za}=\bm{0}$ and its nullity
is $9$, as given by
Proposition~\ref{prop:gaines_invertibility}.
The equations remain consistent; the
difficulty is that their multipliers are
not determined uniquely and their binding
block does not decide releases. CLASSO
instead uses $\alpha=\{2,8\}$ and a
two-dimensional directional QP with a
single sign constraint. The same zero-row
structure persists at later points, with
nullity $p-|a|$.

These examples show that the rank issue
occurs in ordinary sign-constrained
regression, rather than only in specially
constructed degenerate problems. The
directional QP supplies a unique local
primal direction without requiring those
binding rows to be independent, including
at the origin, where the same QP over the
full cone gives the initial segment.

\subsection{Usage notes}\label{sec:usage_notes}

Columns of $\bm{X}$ are usually centred and
scaled before fitting, so that the common
$L_1$ penalty has a comparable interpretation
across coefficients. Any constraints expressed
in the original coefficient units must be
transformed consistently. The diabetes examples
use the standardised design described above.

Once a path has been computed successfully,
coefficients at an intermediate $\lambda$
are obtained by affine interpolation.
Held-out losses can then be evaluated without
another optimisation solve. Cross-validation
requires a separate training path for each
fold, as in Section~\ref{sec:cv}; the
full-data path is used only after selecting
$\lambda$. Information criteria also require
the appropriate degrees-of-freedom expression,
which is not developed here.

The active sets and breakpoint information
may be useful in subsequent inference, but
an exact path does not by itself provide
valid post-selection uncertainty estimates.
Any bootstrap or selective-inference analysis
must account for the fitting and selection
procedure being used. We leave that
statistical development to future work.

\section{Conclusions}\label{sec:conc}

We have presented a directional-QP approach to the
LASSO with general linear equality and inequality
constraints. The QP determines which boundary
coefficients move and which weakly active rows
remain binding, and has a unique primal solution
when the Hessian is positive definite, irrespective
of dependence among the binding rows. Its solution
gives the local direction of the constrained-LASSO
path. The construction recovers the LASSO-modified
LARS direction in the unconstrained setting and,
over the full cone, gives the segment leaving the
origin.

The comparison with the linear-system construction
of \citet{Gaines2018} shows why rank-aware linear
algebra alone does not settle the local update.
A singular binding system may remain consistent,
but imposing all its rows as equalities can
prevent an entering coefficient from moving.
The directional QP resolves the required
release decisions before the continuous
quantities are recovered by linear solves.

For numerical implementation, we refine the
QP's proposed state using directional KKT
sign tests and recompute a primal--dual
certificate at each recorded point.
Following the canonical certificate fixes the
division between subgradients and bound
multipliers that one-sided bounds leave free,
so every recorded interior point on the
synthetic suite is a change of primal direction;
coordinates held by their subgradients alone
require both entry margins to be monitored.

On the synthetic suite, the signed relative
objective differences from reference solves lie between
$-9\times10^{-15}$ and $3.9\times10^{-10}$ at the sampled
$\lambda$ for the problems with $n\ge p$; at $\lambda=0$ the signed
unnormalised differences lie between $-1.5\times10^{-7}$ and $2.0\times10^{-10}$. For $n<p$,
the algorithm solves a ridge-augmented
problem, with matched differences ranging from
$-3.6\times10^{-7}$ to $2\times10^{-9}$ over the five
prescribed ridge strengths. Across the synthetic and diabetes benchmarks,
checks of all $580$ recorded segments found no missed interior
event and all checked KKT residuals within the specified
tolerances.
The differences from unregularised reference solutions are
reported separately and can be substantially larger than
the matched numerical differences.

On the six diabetes configurations, CLASSO
uses $30$--$84$ times fewer LP/QP calls
than the specified $1000$-point sweep and is
$4$--$20$ times faster.
These ratios refer to that baseline, rather
than to a comparison at matched path accuracy.
For the four synthetic scaling configurations,
recorded point counts grow approximately
linearly with $p$ up to $5000$.
Dense factorisations remain the principal
computational limitation; updating them and
warm-starting the directional QP may reduce
the cost.

We prove that the primal path is continuous
and piecewise affine with finitely many
kinks, assuming only a positive-definite
Hessian and a nonempty feasible set.
The local continuation and safe-hold results
likewise require no constraint qualification.
Finiteness of the path does not make a
traversal finite: a policy that recovers a
fresh minimum-norm certificate at each kink
need not terminate, as machine-checked
counterexamples show. Following the canonical
certificate along its lexicographic rate does
terminate, within an explicit bound, and
traverses the whole path in exact arithmetic.
The floating-point implementation of this rule
agrees with an exact rational implementation on
all $210$ small test problems; at the end of every step its
propagated certificate agrees with the
recomputed canonical one, and on the reported
benchmarks it never substitutes a fallback for
the rule's own solves. These checks are
evidence, not proof, that the implementation
realises the exact rule.

A further extension is the unregularised
rank-deficient problem. Ridge augmentation
restores strict convexity but changes the
target, so a path-following method that
selects directly among the unregularised
minimisers would be preferable when that
change is unacceptable.
\citet{Schmelzer2026} discusses a
minimum-$L_2$-norm selection at each
$\lambda$; tracing that selection is a
separate problem.

We have also not developed degrees of
freedom or post-selection inference.
The constrained-fit results of
\citet{Hu2015} and the budget-indexed
analysis of \citet{Schmelzer2026} provide
related starting points, but penalty and
budget indexing must be distinguished
because their relation depends on the
response.

\section*{Acknowledgments}

We thank Will Handley for some useful early discussions of the
constrained LASSO problem.

The authors used Claude (Anthropic; Fable 5.1 and Opus 5) and ChatGPT
(OpenAI;
GPT-6 Astra) for checking analytical calculations, writing code for
numerical computations, literature search, and drafting suggestions;
all results were independently verified, and the authors take full
responsibility for the content.

\section*{Data availability}

The diabetes data are publicly available and were used
in the standardised form distributed with
\texttt{scikit-learn}. The remaining data are generated
by seeded code. The released implementation includes
the stress-suite generators, scaling scripts and
machine-readable results underlying
Tables~\ref{tab:stress_path},
\ref{tab:scalability} and~\ref{tab:nondegeneracy}.
The scaling records specify the machine, core count,
node exclusivity and job identifier for each run.
The Lean sources, pinned toolchain and dependency
manifest accompany the same release.%
\REPOFOOTNOTE

%%%%%%%%%%%%%%%%%%% BIBLIOGRAPHY %%%%%%%%%%%%%%%%%%%

\CILBIBLIOGRAPHY

\appendix

\section{Finite-precision mechanisms in constrained path-following}\label{sec:safeguards}

This appendix summarises numerical failure mechanisms that a
constrained path-following implementation must guard against,
and how this implementation does so. The examples illustrate
why tolerances for solves, state classification and certificate
acceptance must be chosen consistently.
Their occurrence is implementation-dependent; the
observed remedies do not constitute a general numerical
stability or termination proof.

\begin{enumerate}
\item \textit{Missing multiplier rates.}
A rate solve can fail when its inputs inherit
error from the directional solve. Continuing with
$\dot{\bm{\mu}}=\bm{0}$ would give incorrect
subgradient rates and event times. The
implementation substitutes a stated fallback,
reports it, and checks the resulting certificate
(Section~\ref{sec:follow}).

\item \textit{Dual recovery at a perturbed point.}
At an event, rounding can place an inactive
subgradient slightly outside its bound, so a
certificate problem with the propagated
coefficients fixed may be infeasible. The
fixed-state solve (\ref{eq:kink_kkt}) corrects
the primal point first, and the canonical
certificate's bounds are relaxed by the
acceptance tolerance; the bounds are then
checked explicitly.

\item \textit{Accumulated propagation error.}
A direction computed to $10^{-5}$ and
propagated without refinement accumulates
error over successive segments. In a
condition-number-$10^3$ example, the
active residual reached $3.5\times10^{-4}$
by the sixteenth point, and division by
$\lambda=2.5\times10^{-3}$ amplified
this to a subgradient error of $0.14$,
causing a spurious entry. The fixed-state
solve reduces this drift, while certificate
checks use the scaled quantity
$\lambda\bm{g}$.

\item \textit{A held-coordinate crossing is omitted.}
A coordinate held in $b$ may reach the
opposite subgradient boundary. A scan
restricted to $\bar\alpha$ omits this
event and can propagate through a genuine
primal kink. Event \textbf{E5} tests
the second margin explicitly.

\item \textit{A release is misclassified.}
A component of order $10^{-4}$ within a
direction of norm $10^2$ can be assigned
incorrectly by a simple hold/release
threshold. The KKT sign refinement in
Section~\ref{sec:release_refinement}
tests the proposed state and recomputes
its direction after changing an
inconsistent decision.

\item \textit{One-sided bounds are not represented consistently.}
The multiplier of a bound on a zero
coefficient is absent from the
active-support equations, so the split between
it and the subgradient is not determined. Left
free, it produces events at which the path does
not bend; set to zero, a forbidden-side excess
appears to require entry. The canonical
certificate fixes the split
(Section~\ref{sec:one_sided}).

\item \textit{Different certificates give different boundary membership.}
With dependent rows, one valid certificate
may give $|g_j|=1$ and another
$|g_j|=1-10^{-5}$, so an entering coordinate
can fail the numerical boundary test and its
event recur. The implementation retains the
entering coordinate in the candidate boundary
set and then verifies the resulting state.

\item \label{mech:threshold}\textit{A classification threshold
is finer than the available accuracy.}
A boundary threshold below the error in
$\bm{g}$ makes state selection sensitive
to the solver output. At
$\epsilon_{\rm b}=10^{-10}$, two releases
of the same solver gave different paths
on \textsc{none/sumzero}, whose initial
point has two coordinates exactly on the
boundary (Section~\ref{sec:nondegeneracy})
that a computed certificate separated by
$1.95\times10^{-10}$; one path differed by
$32$ per cent in objective value. The adopted threshold
is $\epsilon_{\rm b}=10^{-8}$.
A smaller nominal tolerance therefore
need not give a more accurate path.
\end{enumerate}

The implementation also retains several recovery
measures: minimum-norm choices for non-unique duals;
removal of redundant all-zero inequality rows;
a ridge of $10^{-8}\max_iH_{ii}$ in the reduced
directional QP, whose effect is removed by a
successful fixed-state refinement; a full-space
directional fallback if the nullspace calculation
fails; and the fixed-state certificate when no
canonical certificate can be computed, reported as
in Section~\ref{sec:follow}.
These measures are distinct from the
full-problem ridge used for $n<p$.
The returned traversal status distinguishes
a run reaching $\lambda_{\min}$ from
one stopped by the iteration cap or a
vanishing step. It is reported separately from the
count of points that failed certification, since a run
can reach $\lambda_{\min}$ with certification failures
recorded along the way.

\section{Proof of the local continuation proposition}\label{sec:branch_uniqueness}

\noindent

We prove local continuation by constructing a
certificate along the candidate segment, then
use uniqueness to establish the safe-hold
criterion. The only multiplier-existence
argument required is the finite-dimensional
Farkas result for the polyhedral directional
cone.

\smallskip\noindent\textit{Step 1: the directional multipliers exist.}
Let $\dot{\bm{\beta}}$ minimise the directional problem over the cone
of (\ref{eq:dotprobdef_constraints}). Collect that cone as a finite
row system: equality rows $\bm{A}_{e:}$; the rows
$\mu_i\bm{A}_{i:}$, which impose zero velocity on strongly active
inequalities; unit rows $\bm{e}_j^{\rm t}$ for inactive coordinates
with $|g_j|<1$; together with the inequality rows $-\bm{A}_{i:}$ for
binding constraints and $g_j\bm{e}_j^{\rm t}$ for $\beta_j=0$. At a
minimiser of a quadratic over a polyhedral cone every feasible
displacement has nonnegative directional derivative, so by a finite
Farkas argument the negative gradient
$-(\bm{H}\dot{\bm{\beta}}+\bm{g})$ lies in the cone generated by the
equality rows with free sign and the \emph{active} inequality rows
with nonnegative sign. Reading off those multipliers as
$(\dot{\bm{\nu}},\dot{\bm{\mu}},\bm{q})$ gives
\[
\bm{H}\dot{\bm{\beta}}+\bm{q}
+\bm{A}_{e:}^{\rm t}\dot{\bm{\nu}}
+\bm{A}_{i:}^{\rm t}\dot{\bm{\mu}}=\bm{0},
\]
with $\dot\mu_i=0$ on inactive rows, $\dot\mu_i\le0$ where $\mu_i=0$,
$q_j=g_j$ where $\beta_j\ne0$, and $g_jq_j\ge1$ on the boundary set,
together with the two complementarity relations
$\dot\mu_i\,(\bm{A}_{i:}\dot{\bm{\beta}})=0$ on every inequality row
and $\dot\beta_j\,(q_j-g_j)=0$ at every zero coordinate, which are
what preserve complementarity under the propagation of Step~2. Here
$\bm{q}$ is the rate of $\lambda\bm{g}$, not of $\bm{g}$. Only
symmetry of $\bm{H}$ is used here, not definiteness.

\smallskip\noindent\textit{Step 2: propagate.} For $\Delta\ge0$ set
\[
\bm{\beta}(\Delta)=\bm{\beta}-\Delta\dot{\bm{\beta}},\quad
\bm{\mu}(\Delta)=\bm{\mu}-\Delta\dot{\bm{\mu}},\quad
\bm{\nu}(\Delta)=\bm{\nu}-\Delta\dot{\bm{\nu}},
\]
\[
\bm{g}(\Delta)=\frac{\lambda\bm{g}-\Delta\bm{q}}{\lambda-\Delta}.
\]
Subtracting $\Delta$ times the directional stationarity of Step 1 from
the stationarity at the kink gives stationarity at $\lambda-\Delta$
exactly; the division defining $\bm{g}(\Delta)$ is safe because
$\Delta<\varepsilon\le\lambda$. The feasibility and sign conditions hold initially and their
rates satisfy the one-sided conditions obtained in Step~1,
so they persist on a sufficiently short interval.
Complementarity is preserved by the corresponding
zero-factor and directional complementarity relations.
The same relations ensure
$g_j(\Delta)\beta_j(\Delta)=|\beta_j(\Delta)|$.
Thus each certificate condition holds on an interval
$[0,\varepsilon_r)$ with $\varepsilon_r>0$. There are finitely many
such conditions; take $\varepsilon$ to be the least of the
$\varepsilon_r$ and $\lambda$.

\smallskip\noindent\textit{Step 3: conclude.} On $[0,\varepsilon)$ the
propagated quantities form a full KKT certificate for
(\ref{eq:CLASSOprobdef}) at $\lambda-\Delta$. The problem is convex,
so the conditions are sufficient; with $\bm{H}\succ\bm{0}$ the
objective is strictly convex and the minimiser is unique. Hence
$\bm{\beta}(\Delta)$ is the unique global minimiser, which is
Proposition~\ref{prop:local_continuation}. If two certificates at the
same $(\bm{\beta},\lambda)$ gave different directions, each would
generate such a segment and uniqueness of the minimiser at
$\lambda-\Delta$ would force the segments, hence the directions, to
agree.

\smallskip\noindent\textit{Step 4: safe holds.} Appending
$\dot{\bm{\beta}}_{\hat b}=\bm{0}$ shrinks the feasible cone. If
$\dot\beta_j=0$ for all $j\in\hat b$ then $\dot{\bm{\beta}}$ remains
feasible for the restricted problem and still minimises the same
objective over a smaller set, so it remains a minimiser; strict
convexity makes it the only one. Conversely if $\dot\beta_j\ne0$ for
some $j\in\hat b$ then $\dot{\bm{\beta}}$ is infeasible for the
restricted problem, whose minimiser is therefore a different point and
whose segment is, by Step 3, not the path. Finally, if $\beta_j=0$ and $\dot\beta_j=0$, then
$\beta_j(\Delta)=0$ throughout the segment, proving
Proposition~\ref{prop:safe_hold}. The example following
that proposition illustrates the objective error caused
by imposing an inadmissible hold.

\smallskip\noindent
These arguments, together with the global results of
Section~\ref{sec:theory}, have been formalised and machine-checked in
Lean~4 against Mathlib; the development accompanies this paper and
reports no admitted proofs.

\end{document}